\documentclass[lettersize, journal]{IEEEtran}
\usepackage{cite}
\usepackage{pgfplots}

\pgfplotsset{compat=1.10}

\usepgfplotslibrary{fillbetween}

\usepackage{bm}
\usepackage{lipsum}
\usepackage{makeidx}
\usepackage{enumerate}
\usepackage{color}
\usepackage{cite}
\usepackage{amsmath,amsthm}   
\usepackage{amssymb}
\usepackage{nomencl}
\usepackage{multirow}
\usepackage{graphicx}
\usepackage{epstopdf}
\usepackage{threeparttable}
\usepackage{multicol}
\DeclareGraphicsExtensions{.pdf,.jpeg,.png,.jpg,.emf,.eps}
\usepackage{calc}
\usepackage{here}
\usepackage{url}

\usepackage{upref}
\usepackage{comment}
\usepackage{times}
\usepackage{dsfont}
\usepackage{epic,eepic}
\usepackage{rawfonts}
\usepackage[T1]{fontenc}
\usepackage{latexsym}
\usepackage{amsfonts}

\usepackage{capitalgreekitalic}

\usepackage[font=small,labelfont=bf]{caption}
\usepackage[font=small,labelfont=bf]{subcaption}

\usepackage{xcolor}
\usepackage{tikz,pgfplots}
\usetikzlibrary{calc}
\usetikzlibrary{intersections}
\usetikzlibrary{arrows,shapes}
\usetikzlibrary{positioning}

\newtheorem{corollary}{Corollary}

\newtheorem{lemma}{Lemma}

\theoremstyle{definition}

\allowdisplaybreaks

\begin{document}
\title{Optimization of Rate-Splitting Multiple Access in Beyond Diagonal RIS-assisted 
 URLLC Systems}
\author{Mohammad Soleymani, \emph{Member, IEEE},  
Ignacio Santamaria, \emph{Senior Member, IEEE}, \\
Eduard Jorswieck, \emph{Fellow, IEEE}, and Bruno Clerckx, \emph{Fellow, IEEE}
 \\ \thanks{ 
Mohammad Soleymani is with the Signal and System Theory Group, Universit\"at Paderborn, 33098 Paderborn, Germany (email: \protect\url{mohammad.soleymani@uni-paderborn.de}).  

Ignacio Santamaria is with the Department of Communications Engineering, University of Cantabria, 39005 Santander, Spain (email: \protect\url{i.santamaria@unican.es}).

Eduard Jorswieck is with the Institute for Communications Technology, Technische Universit\"at Braunschweig, 38106 Braunschweig, Germany
(e-mail: \protect\url{jorswieck@ifn.ing.tu-bs.de})

Bruno Clerckx is with the Department of Electrical and Electronic Engineering,
Imperial College London, London SW7 2AZ, U.K and with Silicon Austria Labs (SAL), Graz A-8010, Austria (e-mail: \protect\url{b.clerckx@imperial.ac.uk}; 
\protect\url{bruno.clerckx@silicon-austria.com}).

The work of Ignacio Santamaria was funded by MCIN/ AEI /10.13039/501100011033 under Grant PID2022-137099NB-C43 (MADDIE).  
The work of Eduard A. Jorswieck was partly supported by the Federal Ministry of Education and Research (BMBF, Germany) through the Program of ``Souver\"an. Digital. Vernetzt.'' Joint Project 6G-Research and Innovation Cluster (RIC) under Grant 16KISK031.
}}
\maketitle
\begin{abstract}
This paper proposes a general optimization framework for rate splitting multiple access (RSMA) in beyond diagonal (BD) reconfigurable intelligent surface (RIS) assisted ultra-reliable low-latency communications (URLLC) systems. {
This framework can provide a suboptimal  solution for a large family of optimization problems in which the objective and/or constraints are linear functions of the rates and/or energy efficiency (EE) of users.} Using this framework, we show that RSMA and RIS can be mutually beneficial tools when the system is overloaded, i.e., when the number of users per cell is higher than the number of base station (BS) antennas. Additionally, we show that the benefits of RSMA increase when the packets are shorter and/or the reliability constraint is more stringent.  Furthermore, we show that the RSMA benefits increase with the number of users per cell and decrease with the number of BS antennas. Finally, we show that RIS (either diagonal or BD) can highly improve the system performance, and BD-RIS outperforms regular RIS. 
\end{abstract}

\section{Introduction}
The sixth generation (6G) of communication systems should be around 100 times more reliable than 5G networks and at the same time, provide around 10 times lower latency comparing to 5G networks \cite{wang2023road, gong2022holographic}. Moreover, it is expected that 6G networks become around 100 times more energy efficient and around 10 times more spectral efficient  than 5G systems \cite{wang2023road, gong2022holographic}. To fulfill such ambitious goals, 6G should employ promising technologies such as reconfigurable intelligent surface (RIS) and rate splitting multiple access (RSMA)  
\cite{wu2021intelligent, di2020smart, mao2022rate, clerckx2023primer}. 
In this paper, we propose a general optimization framework to improve spectral efficiency (SE) and energy efficiency (EE) of ultra-reliable and low-latency communication (URLLC) systems by employing RIS and RSMA. 
Moreover, we investigate whether and how RIS and RSMA can be beneficial in URLLC systems, and show how their possible benefits can vary in different operational points depending on the latency and reliability constraints.

\subsection{Related works}\label{sec-i-a}
To support low latency, we cannot operate in very large packet length regimes and have to employ shorter packet lengths in which the Shannon rates are not accurate anymore. 
In \cite{polyanskiy2010channel}, it was shown that the rate for single-input single-output (SISO) point-to-point channels transmitting Gaussian signals can be approximated as
\begin{equation}\label{na}
{r}=
C
-
{Q^{-1}(\epsilon)\sqrt{\frac{V}{n_t}}},
\end{equation}
where $C$ is the Shannon rate, $n_{t}$ is the packet length in bits, $Q^{-1}$ is the inverse of the Gaussian Q-function, $\epsilon$ is the decoding error probability of the message, and $V$ is the channel dispersion. 
The finite-block-length (FBL) rate approximation in \eqref{na} is known as the normal approximation (NA). The accuracy of the NA at different operational points has been vastly discussed in  \cite{polyanskiy2010channel, erseghe2016coding, erseghe2015evaluation}. 
As can be easily verified from \eqref{na}, the FBL rates are smaller than Shannon rates. In addition, as expected, the shorter the packet length is and/or the more stringent the reliability constraint is, the lower rate can be achieved. 
Indeed, we have to transmit at a lower rate to ensure a more reliable communication with low latency. 
Resource allocation and transmission schemes for FBL regimes based on the NA have been studied in \cite{ nasir2020resource,  ghanem2020resource, nasir2021cell}. 
In \cite{ nasir2020resource}, the authors proposed power optimization and beamforming schemes for  a broadcast channel (BC). 
In \cite{ghanem2020resource}, the authors proposed schemes to maximize the weighted sum rate of a multiple-input single-output (MISO) orthogonal frequency division multiple access (OFDMA) URLLC system. Moreover, the paper \cite{nasir2021cell} proposed schemes to maximize the minimum rate and minimum EE of users in a cell-free massive multiple-input multiple-output (MIMO) URLLC system.

One of the targets of 6G is to significantly improve the SE and EE, which can be even more important in URLLC systems. A promising technology that will enable us to meet this target is RIS, which has been shown to enhance the performance of various interference-free and interference-limited systems with Shannon rates and/or the NA in \eqref{na} \cite{wu2021intelligent, di2020smart, wu2019intelligent,  huang2019reconfigurable,     soleymani2022noma, pan2020multicell, soleymani2022improper,  soleymani2022rate, soleymani2023noma, soleymani2023energy, santamaria2023icassp, li2021aerial,  xie2021user, abughalwa2022finite, soleymani2023spectral, vu2022intelligent, ren2021intelligent,   zhang2021irs,   almekhlafi2021joint,  
ranjha2020urllc, ranjha2021urllc, 
 dhok2021non, ghanem2021joint}. 
{
In \cite{wu2019intelligent,  huang2019reconfigurable},  it was shown that RIS can improve the SE and/or EE of MISO BCs. In \cite{soleymani2022noma}, a non-orthogonal multiple access (NOMA) based scheme was proposed for a multi-cell MISO RIS-assisted BC, and it was shown that RIS can significantly increase the minimum rate and EE of the network. 
The superiority of RIS in multi-cell MIMO BCs was studied in \cite{pan2020multicell, soleymani2022improper} when treating interference as noise (TIN) is employed. In \cite{soleymani2022rate, soleymani2023noma, soleymani2023energy}, it was shown that RIS can highly enhance the SE and EE of a multi-cell MIMO BC when the transceivers suffer from I/Q imbalance. In \cite{santamaria2023icassp}, it was shown that RIS can decrease the interference leakage of $K$-user MIMO interference channels.

All the aforementioned papers \cite{wu2019intelligent,  huang2019reconfigurable,     soleymani2022noma, pan2020multicell, soleymani2022improper,  soleymani2022rate, soleymani2023noma, soleymani2023energy, santamaria2023icassp} on RIS considered Shannon rates.   
The use of RIS in FBL regimes has been investigated in
\cite{ 
li2021aerial,  xie2021user, abughalwa2022finite, soleymani2023spectral, vu2022intelligent, ren2021intelligent,   zhang2021irs,   almekhlafi2021joint, ranjha2020urllc, ranjha2021urllc,  dhok2021non, ghanem2021joint}. 
 The authors in \cite{li2021aerial} showed that RIS can improve the performance of URLLC unmanned aerial vehicle (UAV).
In \cite{xie2021user, abughalwa2022finite, soleymani2023spectral}, it was shown that RIS can improve the performance of SISO/MISO BCs with TIN in URLLC systems. 
The authors in \cite{vu2022intelligent} showed that RIS and NOMA can decrease the average decoding rate and at the same time increase the throughput of a two-user SISO BC. 
In \cite{ren2021intelligent}, it was shown that RIS can enhance the performance of a MISO point-to-point URLLC system with a factory automation scenario. Finally, \cite{ghanem2021joint} showed that RIS can increase the sum rate of a multi-cell MISO orthogonal frequency division multiplexing (OFDM) BC. 
For a more detailed review on RIS, we refer the reader to \cite{wu2021intelligent, di2020smart}.}

There are different RIS technologies and architectures. In most of the early works, the matrix modeling the RIS coefficients is assumed to be  diagonal. However, the system performance can be improved by relaxing the diagonality assumption and using beyond diagonal (BD) RIS architectures  \cite{li2023reconfigurable}. In a BD-RIS, each RIS element can be connected to other elements through a circuit \cite{shen2021modeling, li2022beyond}. There are three different possibilities for BD-RIS based on the connectivity of RIS elements: single-connected, group-connected and fully-connected architectures \cite{li2022beyond}. 
Indeed, a regular passive RIS can be considered as a special case of BD-RIS, which can be referred to as single-connected architecture.  In a fully-connected BD-RIS, all the BD-RIS elements are connected to each other, while in a group-connected BD-RIS, each element is connected to only a group of elements, which reduces the implementation complexities.  
The superiority of BD-RIS over a regular RIS has been studied in \cite{li2022reconfigurable, wang2023dual, li2022beyond, santamaria2023snr, ruiz2023low}. 
For instance, in \cite{li2022reconfigurable}, the authors proposed a scheme that results in a non-diagonal phase shift RIS matrix, showing that BD-RIS can improve the performance of single- and multi-user MISO BCs. 
Moreover, it was shown in \cite{wang2023dual} that BD-RIS can outperform RIS in
dual-function radar-communication
 systems. Additionally, \cite{li2022beyond} showed that BD-RIS (with the group- and fully-connected architectures) can enhance sum rate of a MISO BC.
{
It is shown in \cite{santamaria2023snr} that BD-RIS can increase the signal-to-noise ratio (SNR) of SISO and MISO systems.}

Another promising technology to enhance SE and EE is RSMA, which includes  as particular cases many other technologies such as TIN, NOMA, multicasting, broadcasting, and spatial division multiple access (SDMA) \cite{mao2022rate, clerckx2023primer}. In rate splitting (RS) schemes, there are two types of messages: common and private. Each private message is intended  for only a specific user, while common messages are decoded by all or by a group of users depending the employed RSMA scheme \cite{mao2022rate}. Indeed, there are different RSMA schemes based on the number/format of common messages. The simplest RSMA scheme is the 1-layer RS, which is very practical and efficient \cite{hao2015rate, lu2017mmse, clerckx2019rate, mao2020beyond, flores2020linear, soleymani2023rate}.
 In 1-layer RS, there is only one common message, which is decoded by all the users, while treating the private messages as noise. Moreover, each user decodes its own private message after decoding and canceling the common message from the received signal.
 Note that when interference is very weak, the optimal strategy for sum rate maximization or in terms of the generalized degrees of freedom is to treat the interference as noise \cite{annapureddy2009gaussian, geng2015optimality}. Furthermore, if interference is strong, the interfering signal should be decoded and canceled from the received signal, which is widely known as successive interference cancellation (SIC) \cite{sato1981capacity}. RSMA bridges TIN and SIC, which makes RSMA very flexible and powerful. Based on the interference level at users, RSMA can switch between TIN and SIC, without the need to order users, thus reducing the design complexities. 
For a more detailed overview on RSMA, we refer the reader to \cite{mao2022rate, clerckx2023primer}. 

\begin{table*}
\centering
\footnotesize
\caption{A brief comparison of the most related works.}\label{table-1}
\begin{tabular}{|c|c|c|c|c|c|c|c|c|c|c|c|c|c|c|}
	\hline
&RIS
&FBL&
MU communications&Multiple-antenna Systems&
EE metrics&RSMA&SE metrics&BD-RIS
 \\
\hline
  This paper&$\surd$
&$\surd$&$\surd$&$\surd$
&$\surd$&$\surd$&$\surd$&$\surd$
\\
\hline
\cite{ nasir2020resource,  ghanem2020resource}&
&$\surd$&$\surd$&$\surd$
&&&$\surd$&
\\
\hline
\cite{nasir2021cell}&
&$\surd$&$\surd$&$\surd$
&$\surd$&&$\surd$&
\\
\hline
\cite{    pan2020multicell} 
&$\surd$
&&$\surd$&$\surd$
&&&$\surd$&
\\
\hline
\cite{huang2019reconfigurable, soleymani2022improper, soleymani2022noma}&$\surd$&
&$\surd$&$\surd$
&$\surd$&&$\surd$&
\\
\hline
\cite{soleymani2022rate}&$\surd$
&&$\surd$&$\surd$
&$\surd$&$\surd$&$\surd$&
\\
\hline
\cite{ li2021aerial, ghanem2021joint }&$\surd$
&$\surd$&$\surd$&$\surd$
&&&$\surd$&
\\
\hline
\cite{ vu2022intelligent, xie2021user,     almekhlafi2021joint}&$\surd$
&$\surd$&$\surd$
&&&&$\surd$&
\\
\hline
\cite{ ren2021intelligent,   zhang2021irs}&$\surd$
&$\surd$&&$\surd$
&&&$\surd$&
\\
\hline
\cite{abughalwa2022finite}&$\surd$
&$\surd$&$\surd$&
&&&$\surd$&
\\
\hline
\cite{li2022robust, wang2023flexible, xu2022max, xu2022rate, ou2022resource, vu2022short}&
&$\surd$&$\surd$&$\surd$
&&$\surd$&$\surd$&
\\
\hline
\cite{singh2023rsma}&$\surd$
&$\surd$&$\surd$&&
&$\surd$&$\surd$&
\\
\hline
\cite{liu2022network, karacora2022rate, singh2022performance}&
&$\surd$&$\surd$&&
&$\surd$&$\surd$&
\\
\hline
\cite{katwe2022rate}&$\surd$
&$\surd$&$\surd$&$\surd$
&$\surd$&$\surd$&
&
\\
\hline
\cite{dhok2022rate}&$\surd$&$\surd$&$\surd$&&&$\surd$&$\surd$&
\\
\hline
		\end{tabular}
\normalsize
\end{table*}

The performance of RSMA with FBL has been studied in \cite{  wang2023flexible, xu2022max, xu2022rate, ou2022resource,  li2022robust, vu2022short,  liu2022network, karacora2022rate, singh2022performance} for systems without RIS. 
In \cite{wang2023flexible}, the authors proposed a flexible RSMA scheme for a MISO BC with FBL and showed that their proposed scheme outperforms SDMA and NOMA. In \cite{xu2022max, xu2022rate}, it was shown that RSMA can achieve the same minimum or sum rate as NOMA and SDMA with a smaller packet length, meaning that RSMA may reduce latency for a  given target rate. In \cite{ou2022resource}, the authors proposed resource allocation schemes for RSMA in a MISO URLLC BC, showing that RSMA provides a higher effective throughput than NOMA. Additionally, they showed that RSMA can reduce latency and enhance reliability. 
{
In \cite{li2022robust}, a robust beamforming design was proposed for a MISO BC with 1-layer RS, and it was shown that RSMA can significantly increase the minimum rate of the users in the considered scenario. 
In \cite{singh2022performance}, it was shown that a 1-layer RS scheme can enhance the performance of a SISO UAV-assisted BC with imperfect channel state information (CSI) and imperfect SIC by considering both infinite and finite block length regimes.}

Finally, we summarize some of the most related works in Table \ref{table-1}. As indicated, RIS and RSMA are powerful tools to improve SE and EE of various systems. 
However, there are only a limited number of papers that considered the performance of RSMA in RIS-assisted URLLC systems (i.e., \cite{ singh2023rsma, dhok2022rate, katwe2022rate}). 
{
In \cite{singh2023rsma}, 1-layer RS schemes were proposed for a SISO system aided by RIS and UAV in which RIS is mounted at a UAV relay that operates in a full-duplex mode with the decode and forward strategy. Then it was shown that RSMA can outperform NOMA and OMA schemes in both infinite and finite block length regimes.   The authors in \cite{dhok2022rate} proposed 1-layer RS schemes for both infinite and finite block length regimes in a SISO STAR-RIS-assisted BC with spatially-correlated channels. In \cite{katwe2022rate}, it was shown that RSMA can increase the global EE of a MISO RIS-assisted BC with FBL regime.}

To sum up, to the best of our knowledge, the paper \cite{katwe2022rate} is the only paper on RSMA in multiple-antenna RIS-assisted URLLC systems. 
Moreover, there is only one paper considering EE metrics in URLLC systems with RSMA, i.e., \cite{katwe2022rate}, in which, it was shown that RSMA can increase the global EE of a single-cell MISO RIS-assisted BC. 
Thus, the performance of RSMA in multi-antenna RIS-assisted systems should be further studied. Additionally,  more energy-efficient RSMA schemes should be developed for URLLC systems.

\subsection{Motivation}
RSMA is a very effective and flexible tool to manage interference that encompasses a large variety of multiple-access technologies such as NOMA, TIN, SDMA, broadcasting and multi-casting \cite{mao2022rate,clerckx2019rate}. Moreover, RIS enables optimizing environment by modulating channels, which can be employed to neutralize interference and/or to improve coverage. Hence, one might expect that the RSMA benefits are reduced when RIS is employed since RIS can manage interference in some scenarios especially in multiple-antenna systems. 
{
However, in \cite{soleymani2022rate, li2023sum, yang2020energy, shambharkar2022performance, li2023synergizing, li2022rate}, it was shown that RSMA and RIS can be mutually beneficial tools in overloaded 
systems with Shannon rate. Unfortunately, in FBL regimes, the Shannon rates are not accurate, which makes optimizing parameters more complicated and can bring some new challenges/tradeoffs. 
Moreover, the solutions for Shannon rates cannot be used in FBL regime, which further motivates developing  specific RSMA techniques for RIS-assisted URLLC systems. }

{
In this work, we study the role of RIS and RSMA in multiple-antenna URLLC systems, and particularly, provide an answer to the following question: {\em What is the impact of RIS on the performance of RSMA in multiple-antenna URLLC systems?} We show that the impact of RIS on RSMA is a dichotomy depending on  the network load. More specifically, RIS can enhance the benefits of RSMA in overloaded systems, but the RSMA benefits decrease (or even become negligible) by optimizing RIS components in underloaded systems.  Moreover, in this work, we investigate the impact of packet length, which is related to the latency constraint, and the reliability constraint on the RSMA performance. We show that the benefits of RSMA increase when packets are shorter or when the maximum tolerable decoding probability is smaller.
}

\subsection{Contribution}

The  main goal of this work is to investigate the overall system performance as well as the  specific role of RIS and RSMA in URLLC systems. We not only show that RIS and/or RSMA can significantly improve the system performance, but also clarify the role of RIS and RSMA in such improvements.
To this end, 
{we propose an optimization framework for RSMA to enhance SE and EE of MISO (BD-)RIS-assisted URLLC systems, which can be applied to every interference-limited system with 1-layer RS.} As shown in Table \ref{table-1} and discussed in Section \ref{sec-i-a}, there are a limited number of works on RSMA in RIS-assisted URLLC systems, and to the best of our knowledge, the SE of RSMA in multiple-antenna RIS-assisted systems with FBL regimes has not been studied yet.
{Thus, it is required to develop 
a general optimization framework for RSMA in multiple-antenna RIS-assisted URLLC systems  that can provide a solution for a large family of optimization problems, including various SE and EE metrics such as the  minimum weighted rate, weighted sum rate, minimum weighted EE, and global EE.}  
In this paper, we address this issue by proposing a framework to solve every optimization problem
in which the objective and/or constraints are linear functions of the rates and/or EE of users and/or the received powers. 
{
Note that in \cite{soleymani2023spectral}, we proposed resource allocation schemes for STAR-RIS-assisted URLLC systems with TIN. However, in this work, we extend the results in \cite{soleymani2023spectral} to RSMA and consider another technology for RIS (i.e., BD-RIS). Indeed, not only the transmission strategy at the BSs is different from \cite{soleymani2023spectral}, but also the RIS technologies vary in these two papers.}

To clarify the role of RIS/RSMA, we define two operational regimes based on the number of users and the number of BS antennas. We call the system underloaded if the number of BS antennas is higher than the number of users per cell. 
Otherwise, we call the system overloaded. {
We show that  RIS and RSMA are mutually beneficial tools in overloaded systems; however, the benefits of RSMA decrease by employing RIS in underloaded systems.} 
The reason is that   
 the interference level is lower in underloaded systems than in overloaded systems. 
Hence, interference in underloaded systems can be managed in a simpler way by SDMA and optimizing channels through RIS. However, in   overloaded systems, we require more powerful interference-management techniques such as RSMA to mitigate interference. We also show that RIS with TIN may even perform worse than RSMA without employing RIS in an overloaded system, which shows the importance of RSMA. To sum up, the role of RSMA is to manage interference especially in overloaded systems. The role of RIS is mainly to improve the coverage in both underloaded and overloaded systems, as well as to partly manage interference in underloaded systems.   

We, moreover, aim at investigating the impact of the reliability and latency constraints on the performance of RSMA. We show that the benefits of RSMA in overloaded systems increase when the  reliability constraint is more stringent and/or when the packet lengths are shorter. This shows that RSMA can enhance reliability and ensure a lower latency. We also show that the benefits of RSMA increase with the number of users per cell. The reason is that the interference level increases with the number of users per resources, which makes RSMA as an interference-management technique more beneficial. Additionally, we show that the benefits of RSMA decrease with the number of antennas at BSs. When the number of BS antennas increases, indeed the number of resources increases, which makes it easier to manage interference by SDMA. We also show that RIS can significantly improve the EE and SE of the system even with a relatively low number of RIS elements per users.

In this paper, we also develop optimization techniques for  BD-RIS with a group-connected architecture of group size two. To the best of our knowledge, this is the first wok that studies BD-RIS in URLLC systems. 
We consider two feasibility sets for optimizing the BD/diagonal RIS elements and show that 
RIS (either regular or beyond diagonal) can significantly improve the system performance. Moreover, we show that BD-RIS with group-connected architecture of group size two can outperform a regular RIS.

\subsection{Paper outline}
This paper is organized as follows. Section \ref{sec-sym} presents the system model and formulates the problem. 
Section \ref{seciii} proposes schemes to optimize the beamforming vectors. 
Section \ref{seciv} provides solutions for optimizing the BD-RIS elements. 
Section \ref{secnum} presents some numerical results. 
Finally, Section \ref{seccon} concludes the paper.  
\section{System model}\label{sec-sym}

\begin{figure}[t!]
    \centering
\includegraphics[width=.48\textwidth]{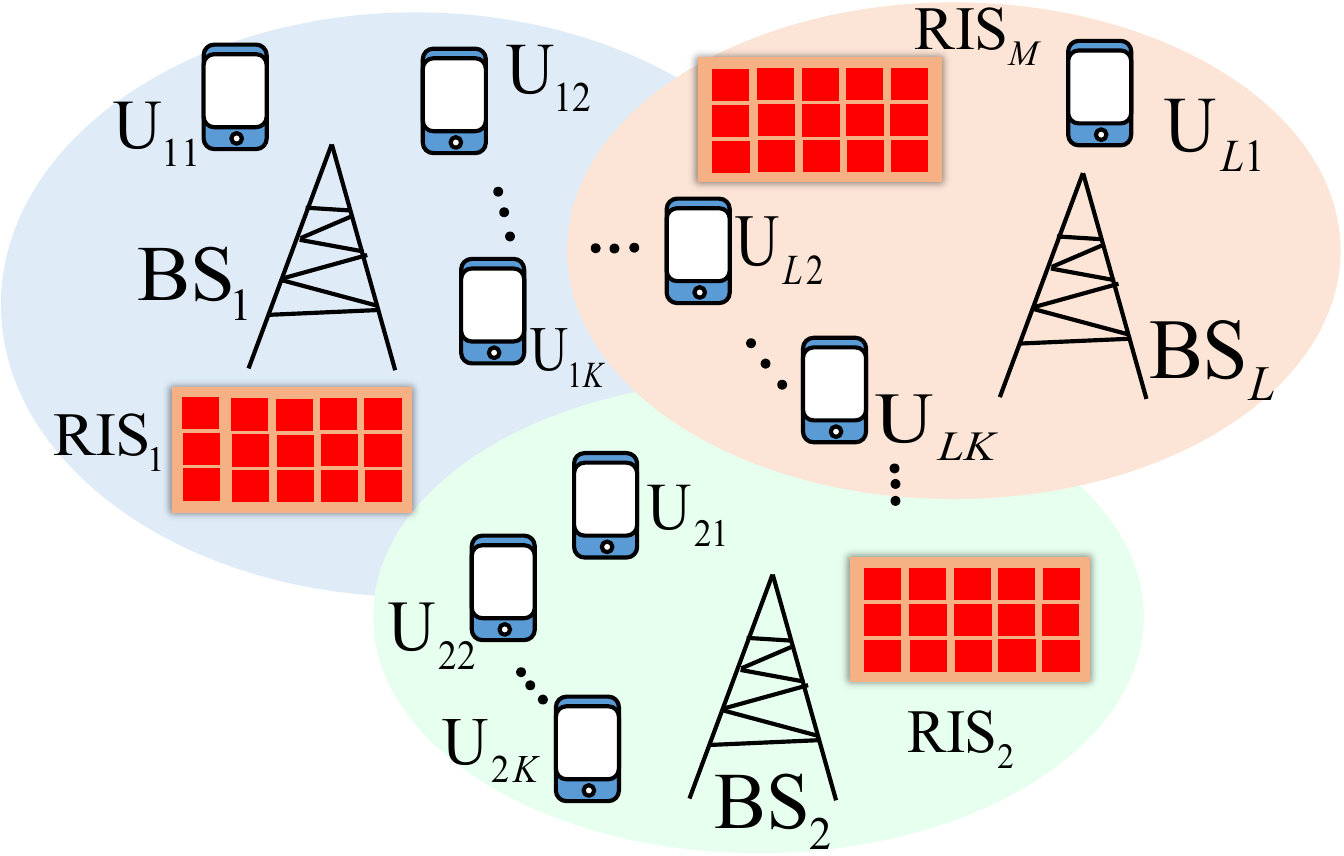}
     \caption{A multicell broadcast channel with RIS.}
	\label{Fig-sys-model}
\end{figure}
{
We propose an optimization framework for 1-layer RS, which can be applied to any interference-limited (BD-)RIS-assisted URLLC system and is able to provide a solution for a large family of  optimization problems in which the objective/utility function and/or constraints are linear functions of the rates and/or EE of users and/or the received power.} 
As an illustrative example, we consider a multicell broadcast channel (BC) with $L$ multiple-antenna base stations (BSs) with $N_{BS}$ transmit antennas, as shown in Fig. \ref{Fig-sys-model}. We assume that each BS serves $K$ single-antenna users, and there are $M\geq L$ BD-RIS and/or regular RISs with $N_{RIS}$ components each. 
{
Moreover, perfect, global and instantaneous CSI is assumed to be available at all transceivers similar to many other works on RIS \cite{ huang2019reconfigurable, wu2019intelligent, soleymani2022noma, pan2020multicell, soleymani2022improper, soleymani2022rate}. 
It should also be emphasized that this is a common assumption when proposing  resource allocation schemes for URLLC systems  \cite{nasir2020resource, wang2023flexible, ghanem2021joint,  ghanem2020resource, soleymani2023spectral}. In this case, it is assumed that the coherence time of the channels is large, and the system remains static for a long time. Thus, after acquiring CSI, the obtained beamforming vectors, RSMA parameters and RIS elements solutions can be used over many time slots. Therefore, the overhead associated with channel estimation (as well as that associated with obtaining the solution) is typically not considered to increase latency. However, such schemes may not be applicable for fast fading systems.
Additionally, it should be noted that considering perfect CSI can provide an upper bound for the benefits of all the schemes. If such benefits are minor  even with perfect CSI, then it may happen that RSMA and RIS cannot be effective in URLLC systems with imperfect CSI.} 
\subsection{RIS model}\label{sec-ris}
In this paper, we consider single-sector BD-RISs in the reflective mode with group-connected architecture of group size two.
In this case, the channel between BS $i$ and user $k$ associated to BS $l$, denoted by u$_{lk}$, is 
\begin{equation}\label{ch}
\mathbf{h}_{lk,i}\!\left(\{\bm{\Theta}\}\right)\!
=\!\underbrace{\sum_{m=1}^M\mathbf{f}_{lk,m}\bm{\Theta}_m
\mathbf{G}_{mi}}_{\text{Links through RIS}}+\!
\underbrace{\mathbf{d}_{lk,i}}_{\text{Direct link}}\!\!\in \mathbb{C}^{1\times N_{BS}},
\end{equation}
where $\mathbf{d}_{lk,i}\in \mathbb{C}^{1\times N_{BS}}$ is the direct link between BS $i$ and u$_{lk}$, $\mathbf{G}_{mi}\in \mathbb{C}^{N_{RIS}\times N_{BS}}$ is the channel matrix between BS $i$ and BD-RIS $m$, $\mathbf{f}_{lk,m}\mathbb{C}^{1\times N_{RIS}}$ is the channel vector between BD-RIS $m$ and u$_{lk}$, $\{\bm{\Theta}\}=\{\bm{\Theta}_m\}$ is the set containing all the BD-RIS components. For a regular RIS, $\bm{\Theta}_m$ is a diagonal matrix, given by
$\bm{\Theta}_m
=\text{diag}\left(\theta_{m_1}, \theta_{m_2},\cdots,\theta_{m_{N_{RIS}}}\right)$, 
where $\theta_{m_i}$ is the coefficient corresponding to the $i$-th element of RIS $m$.
However, in a BD-RIS, the diagonality assumption is relaxed, and  $\bm{\Theta}_m$ is a symmetric non-diagonal matrix. 
For BD-RIS with group-connected architecture of group size two, $\bm{\Theta}_m$ is a block-diagonal matrix as 
\begin{align}
\bm{\Theta}_m=\text{diag}(\bm{\Theta}_{m_1},\bm{\Theta}_{m_2},\cdots,\bm{\Theta}_{m_G}),\,\,\,\bm{\Theta}_{m_g}=\bm{\Theta}_{m_g}^T,
\end{align}
where $\bm{\Theta}_{m_g}$ for all $m,g$ is a 2-by-2 symmetric matrix, and $G=\frac{N_{RIS}}{2}$. {
Note that without loss of generality, we assume that $N_{RIS}$ is an even number since it is assumed that each RIS element is connect to another RIS element. If $N_{RIS}$ is an odd number, there is one RIS element, which is single connected, and our analysis can easily include this case as well.}

 There can be two different constraints for the symmetric matrices $\bm{\Theta}_{m_g}$s. First,
we have the convex constraint $\bm{\Theta}_{m_g}\bm{\Theta}_{m_g}^H\preceq{\bf I}$ for all $m,g$, which results in the following feasibility set:
 \begin{align} 
\mathcal{T}_U=\{\bm{\Theta}_{m_g}=\bm{\Theta}_{m_g}^T,\bm{\Theta}_{m_g}\bm{\Theta}_{m_g}^H\preceq{\bf I},\forall m,g\}
\end{align}
Second, we have 
$\bm{\Theta}_{m_g}\bm{\Theta}_{m_g}^H={\bf I}$, 
which yields 
\begin{align} 
\mathcal{T}_I=\{\bm{\Theta}_{m_g}=\bm{\Theta}_{m_g}^T,\bm{\Theta}_{m_g}\bm{\Theta}_{m_g}^H={\bf I},\forall m,g\}
\end{align}
Note that $\mathcal{T}_I\subset \mathcal{T}_U$. Indeed, $\mathcal{T}_U$ includes $\mathcal{T}_I$ as a special case and should not perform worse than $\mathcal{T}_I$.
Note that if u$_{lk}$ (or BS $i$) is not in the reflection space of RIS $m$, then we have $\mathbf{f}_{lk,m}=\mathbf{0}$ (or $\mathbf{G}_{mi}={\bf 0}$). In other words, in order to get a signal through a reflective BD-RIS, the transceivers should be in the reflection space of the BD-RIS. For more details on different architectures of BD-RIS, we refer the reader to \cite{li2022beyond}. 
{
It should be noted that realizing a BD-RIS is more challenging than a regular RIS since the BD-RIS elements should be connected through a proper circuit design. Additionally, optimizing the BD-RIS elements is more complicated than optimization of a regular RIS due to the unitary and symmetric constraints. However, BD-RIS is more general than a regular RIS, and an optimal BD-RIS scheme never performs worse than any scheme with regular RIS.}
Hereafter, for notational simplicity, we drop the dependency of the channels on $\bm{\Theta}_m$ and represent the channels as $\mathbf{h}_{lk,i}$ for all $i,l,k$.

{
Note that there are more assumptions regarding the coefficients of regular RISs, as mentioned in \cite[Sec. II.B]{wu2021intelligent}. However, to the best of our knowledge, there is no other feasibility set for BD-RIS coefficients (apart from $\mathcal{T}_U$ and $\mathcal{T}_I$) in the literature. Indeed, we are not aware of any implementation of BD-RIS in practice, and the feasibility sets $\mathcal{T}_U$ and $\mathcal{T}_I$ are mainly based on the fact that the BD-RIS is assumed to operate in a passive mode. Such feasibility sets may provide an upper bound for the performance of BD-RIS, which can be useful to clarify the limitations of BD-RIS and/or to investigate how much gain can be obtained by moving beyond the diagonality assumption for the RIS coefficients.
}

\subsection{Signal model}
We consider the 1-layer RS to manage intra-cell interference. In the 1-layer RS, each BS transmits one common message, to be decoded by all its associated users, in addition to $K$ private messages intended for each individual user as 
\begin{equation}
\mathbf{x}_l=
\underbrace{
\mathbf{x}_{c,l}s_{c,l}
}_{\text{Common message}}
+
\underbrace{
\sum_{k=1}^K\mathbf{x}_{p,lk}s_{p,lk}
}_{\text{Private messages}}
\in\mathbb{C}^{N_{BS}\times 1},
\end{equation}
 where $s_{c,l}\sim\mathcal{CN}\left(0,1\right)$ is the common message of BS $l$, and $s_{p,lk}\sim\mathcal{CN}\left(0,1\right)$ is the private message intended for u$_{lk}$. 
Moreover, $\mathbf{x}_{c,l}$ and $\mathbf{x}_{p,lk}$ are, respectively, the beamforming vectors corresponding to the common message $s_{c,l}$  and private message $s_{p,lk}$. 
Note that $s_{c,l}$ and $s_{p,lk}$ for all $l,k$ are independent and identically distributed 
proper Gaussian signals. {Moreover, the transmission power of the common message at BS $l$ (private message intended for u$_{lk}$) is $\mathbf{x}_{c,l}^H\mathbf{x}_{c,l}$ ($\mathbf{x}_{p,lk}^H\mathbf{x}_{p,lk}$).}

The received signal for the user u$_{lk}$ is
\begin{align}\label{rec-sig}
y_{lk}&
\nonumber
=
\underbrace{\mathbf{h}_{lk,l}
\mathbf{x}_{c,l}s_{c,l}}_{\text{Desired C. signal}}+
\underbrace{\mathbf{h}_{lk,l}
\mathbf{x}_{p,lk}s_{p,lk}}_{\text{Desired P. signal}}+
\underbrace{\mathbf{h}_{lk,l}
\sum_{j=1,j\neq k}^{K}\mathbf{x}_{p,lj}
s_{p,lj}}_{\text{Intracell interference}}
\\&\hspace{1cm}
+
\underbrace{ \sum_{i=1,i\neq l}^L\mathbf{h}_{lk,i}
\mathbf{x}_{i}}_{\text{Intercell interference}}+
\underbrace{n_{lk}}_{\text{Noise}},
\end{align}
where 
$\mathbf{h}_{lk,i}\in\mathbb{C}^{1\times N_{BS}}$ is the  channel between BS $l$ and u$_{lk}$, given by \eqref{ch}, and $n_{lk}\sim\mathcal{CN}\left(0,\sigma^2\right)$ is additive  white Gaussian noise, which is independent of the transmitted signals. 

\subsection{Rate and energy-efficiency expressions}
Each user first decodes the common message, treating all other signals as noise. 
Thus, the rate of decoding $s_{c,l}$ at u$_{lk}$ is  \cite{polyanskiy2010channel}, \cite[Eq. (8)]{nasir2021cell}
\begin{equation}\label{rate}
\bar{r}_{c,lk}=
\underbrace{\log\left(1+\gamma_{c,lk}\right)}_{\text{\small Shannon Rate}}
-
\underbrace{Q^{-1}(\epsilon_c)\sqrt{\frac{V_{c,lk}}{n_c}}}_{\delta_{c,lk}(\{\mathbf{x}\},\{\bm{\Theta}\})},
\end{equation}
where $V_{c,lk}$ is the channel dispersion for decoding $s_{c,l}$ at u$_{lk}$, $n_{c}$ is the packet length of the common message in bits, 
$\epsilon_c$ is the decoding error probability of the common message, 
and  $\gamma_{c,lk}$ is the corresponding SINR given by
\begin{equation}\label{sinr}
\gamma_{c,lk}=\frac{|\mathbf{h}_{lk,l}
\mathbf{x}_{c,l}|^2}{\sigma^2+\sum_{i\neq l}
|\mathbf{h}_{lk,i}
\mathbf{x}_{c,i}|^2
+\sum_{ij}|\mathbf{h}_{lk,i}
\mathbf{x}_{p,ij}|^2}.
\end{equation}
The optimal channel dispersion is \cite{polyanskiy2010channel}
\begin{equation}\label{dis-opt}
V_{c,lk}^{opt}=1-\frac{1}{\left(1+\gamma_{c,lk}\right)^2},
\end{equation}
which is not achievable by Gaussian signals in the presence of interference \cite{scarlett2016dispersion}.
An achievable channel dispersion for Gaussian signals in interference-limited systems  is \cite{scarlett2016dispersion}
\begin{equation}\label{dis-}
V_{c,lk}=2\frac{\gamma_{c,lk}}{1+\gamma_{c,lk}}.
\end{equation}
The common message $s_{c,l}$ should be decodable for all the users associated to BS $l$. Hence, the transmission rate of $s_{c,l}$ must be smaller than or equal to the minimum achievable rate of decoding $s_{c,l}$ at the users associated to BS $l$, i.e., 
\begin{equation}\label{fis=}
r_l
\leq \min_{
k}\left\{\bar{r}_{c,lk}\right\}\triangleq r_{c,l}.
\end{equation}

Each user decodes and cancels the common message. After this, it decodes its own private message, treating the remaining signals as noise. Thus, the decoding rate of $s_{p,lk}$ at u$_{lk}$ is  \cite{polyanskiy2010channel}, \cite[Eq. (8)]{nasir2021cell}
\begin{equation}\label{rate}
r_{p,lk}=
\underbrace{\log\left(1+\gamma_{p,lk}\right)}_{\text{\small Shannon Rate}}
-
\underbrace{Q^{-1}(\epsilon_p)\sqrt{\frac{V_{p,lk}}{n_p}}}_{\delta_{p,lk}(\{\mathbf{x}\},\{\bm{\Theta}\})},
\end{equation}
where  $\epsilon_p$ is the decoding error probability of the private message, $n_{p}$ is the packet length of the private message in bits, and  $\gamma_{p,lk}$ is the SINR for decoding $s_{p,lk}$ given by
\begin{equation}\label{sinr=p}
\gamma_{p,lk}=\frac{|\mathbf{h}_{lk,l}
\mathbf{x}_{p,lk}|^2}{\sigma^2\!+\!\sum_{i\neq l}
|\mathbf{h}_{lk,i}
\mathbf{x}_{c,i}|^2\!
+\!\sum_{[ij]\neq [lk]}\!|\mathbf{h}_{lk,l}
\mathbf{x}_{p,lj}|^2
},
\end{equation}
where $
\sum_{[ij]\neq [lk]}|\mathbf{h}_{lk,i}
\mathbf{x}_{p,ij}|^2=\sum_{ij}|\mathbf{h}_{lk,i}
\mathbf{x}_{p,ij}|^2-|\mathbf{h}_{lk,l}
\mathbf{x}_{p,lk}|^2$. 
Similarly, an achievable channel dispersion is \cite{scarlett2016dispersion}
\begin{equation}\label{dis-p}
V_{p,lk}=2\frac{\gamma_{p,lk}}{1+\gamma_{p,lk}}.
\end{equation}
Finally, the rate of u$_{lk}$ is 
\begin{equation}
r_{lk}=r_{p,lk}+r_{c,lk},
\end{equation}
where $r_{c,lk}\geq 0$ is the portion of the rate that u$_{lk}$ gets from the common message. Note that $\sum_kr_{c,lk}\leq r_{c,l}$, where $r_{c,l}$ is given by \eqref{fis=}.

Finally, the EE of u$_{lk}$ is defined as \cite{zappone2015energy}
\begin{equation}
e_{lk}=\frac{r_{lk}}{p_c+\eta\left(\mathbf{x}_{p,lk}^H\mathbf{x}_{p,lk}+\mathbf{x}_{c,l}^H\mathbf{x}_{c,l}/K\right)},
\end{equation}
where $\eta^{-1}$ is the power efficiency of the BSs, and $p_c$ is the constant power consumption to transmit data to each user, given by \cite[Eq. (27)]{soleymani2022improper}. 

\subsection{Discussion on reliability and latency constraints}
The reliability constraint is modeled by the decoding error probabilities, $\epsilon_c$ and $\epsilon_p$. The total decoding error probability for a user can be upper bounded as
\begin{equation} 
\epsilon_t=\epsilon_c+(1-\epsilon_c)\epsilon_p\approx \epsilon_c+\epsilon_p.
\end{equation}
In general, the error probability of decoding private and/or common messages can be different at each user, based on their service requirements. To simplify the notations, we consider a symmetric system in which the decoding error probability of the private messages are the same at all users. However, this framework can be easily modified for asymmetric scenarios. 

The latency constraint can be also translated to a rate constraint. The reason is that, if the latency for a packet with length $n$ bits should be less than $T$ seconds, then its transmission rate should be higher than $r\geq \frac{n}{\beta T}$ (b/s/Hz), where $\beta$ is the used bandwidth. 
Thus, the rate of u$_{lk}$ should be $r_{lk}\geq r_{lk}^{th}=\frac{n_p+n_c}{\beta T}$ (b/s/Hz), where $T$ is the latency constraint. 

Note that $\epsilon_t$ and $r_{lk}^{th}$ are upper bounds for the decoding error probability and the latency, respectively, since it may happen that u$_{lk}$ receives its rate from only the common message \cite{clerckx2019rate}. In this case, $\epsilon_t=\epsilon_c$ and $r_{lk}^{th}=\frac{n_c}{\beta T}$. However, in this paper, we consider the upper bounds to ensure that the latency and reliability constraints are met.

\subsection{Problem statement}
We consider a general optimization problem,  similar to, e.g., \cite[Eq. (29)]{mao2022rate}, as 
\begin{subequations}\label{ar-opt}
\begin{align}
\label{12-a}
 \underset{\{\mathbf{x}\}\in\mathcal{X},\{\bm{\Theta}\}\in\mathcal{T},\mathbf{r}_c
 }{\max}
& 
  f_0\left(\left\{\mathbf{x}\right\},\{\bm{\Theta}\}\right) &\\
\text{s.t.}   \,\,\,\,\,\,\,\,\,\,
&  f_i\left(\left\{\mathbf{x}\right\},\{\bm{\Theta}\}\right)\geq0,
\forall i,
\\
\label{4-c}
 &
r_{lk}\geq r_{lk}^{th},
\forall l,k,
\\\label{4-b}
 &  \sum_{k=1}^K\!\!r_{c,lk}\!\leq\! \min_{k}\!\left\{\bar{r}_{c,lk}\right\}\!\triangleq\! r_{c,l}(\{\mathbf{x}\})\!,
\forall lk,
\\
\label{4-d}
&
r_{c,lk}\geq 0,
\forall l,k,
 \end{align}
\end{subequations}
where $\{\mathbf{x}\}$ is the set of the beamforming vectors, $\mathcal{X}$ is the feasibility set for the beamforming vectors, $\mathbf{r}_c=\left\{r_{c,lk},\forall lk\right\}$ is the set of the common rates, $\mathcal{T}$ is the feasibility set for RIS components, which can be either $\mathcal{T}_I$ or $\mathcal{T}_U$, 
$f_i$s are linear functions of rates/EEs and/or concave/convex/linear functions of beamforming vectors and/or channels,
  \eqref{4-c} is the latency constraint, \eqref{4-b} is the  decodability constraint of each common message,  \eqref{4-d} is because of non-negative rates. {
The variables 
  $\mathbf{r}_c$, $\{\mathbf{x}\}$ and $\{\bm{\Theta}\}$ are the optimization parameters. 
Indeed, we jointly optimize the RSMA parameters, beamforming vectors, transmission powers and RIS elements.}
Finally, the feasibility set $\mathcal{X}$ is
\begin{equation}
\mathcal{X}=\left\{\mathbf{x}_{p,lk},\mathbf{x}_{c,l}:\mathbf{x}_{c,l}^H\mathbf{x}_{c,l}
+\sum_{k}\mathbf{x}_{p,lk}^H\mathbf{x}_{p,lk}\leq p_l,\forall l\right\},
\end{equation}
where $p_l$ is the power budget of BS $l$. 

The general optimization problem \eqref{ar-opt} includes  a large family of optimization problems for enhancing SE and EE of the system. For instance, \eqref{ar-opt} includes  the minimum weighted rate maximization (MWRM), weighted sum rate maximization (WSRM),  minimum weighted EE maximization (MWEEM), global EE maximization (GEEM) problems, to mention a few. 
\vspace{-.6cm}
\section{Proposed optimization framework for optimizing beamforming vectors}\label{seciii}
Our proposed optimization framework is an iterative optimization technique, based on majorization minimization (MM) and alternating optimization (AO). That is, we first fix RIS components to $\{\bm{\Theta}^{(t-1)}\}$ and solve \eqref{ar-opt} over the beamforming vectors $\{\mathbf{x}\}$ to obtain $\{\mathbf{x}^{(t)}\}$. We then alternate the optimization parameters and solve \eqref{ar-opt} for fixed beamforming vectors $\{\mathbf{x}^{(t)}\}$ to obtain $\{\bm{\Theta}^{(t)}\}$. 
We continue this procedure until a convergence metric is met. 
{
To obtain a feasible initial point, we employ the approach in \cite[Appendix A]{soleymani2023spectral} and maximize the minimum of $\gamma_{p,lk}$ and $\gamma_{c,lk}$ for all $l,k$ until the rates $r_{p,lk}$ and $r_{c,lk}$ for all $l,k$ are positive. The reason is that, the rates with FBL can be negative if the SINR is very low as shown in \cite[Lemma 2]{soleymani2023spectral}, which is not a feasible initial point. Note that it is not ensured that AO and/or MM converge to a global optimal solution. However, to the best of our knowledge, there is no work on multi-user and multiple-antenna RIS-assisted systems that obtains the global optimal solution of the general optimization problem in \eqref{ar-opt}, even with a simpler scenario, i.e., considering Shannon rates and/or TIN. Additionally, most of the existing works on RIS (e.g., \cite{wu2019intelligent, huang2019reconfigurable, soleymani2022noma, pan2020multicell}) employ AO to jointly optimize the transmission parameters and RIS elements. Hence, we employ AO to tackle \eqref{ar-opt}.}

In this section, we present our schemes to optimize beamforming vectors as well as the RSMA parameters.
To update $\{\mathbf{x}\}$, we solve  \eqref{ar-opt} for fixed $\{\bm{\Theta}^{(t-1)}\}$, which is a complicated non-convex optimization problem. To this end, we employ an MM-based approach. 
That is, we first find suitable surrogate functions for the rates and then, 
solve the corresponding surrogate optimization problem. 
{
The surrogate functions should be  concave lower bounds for the rates, satisfying the three conditions in \cite[Lemma 3]{soleymani2020improper} to ensure a convergence to a stationary point of \eqref{ar-opt}. Note that our proposed framework for systems without RIS converges to a stationary point of \eqref{ar-opt} since the surrogate functions fulfill these conditions. We will discuss the optimality conditions of our framework for systems with (BD-)RIS in the next section.}
In the following lemma, we present concave lower bounds for $r_{p,lk}$ and $\bar{r}_{c,lk}$ that satisfy the conditions in \cite[Lemma 3]{soleymani2020improper}.
\begin{lemma}\label{lem-1}
Concave and quadratic lower bounds for $r_{p,lk}$ and $\bar{r}_{c,lk}$ are, respectively, 
\begin{multline}\label{eq=lem=q}
\!\!\!\!r_{p,lk}\!\!\geq\!\! \tilde{r}_{p,lk}\!\!=\!\!
a_{p,lk}\!
+\!
{2c_{lk}\mathfrak{R}
\left\{
\left(
\mathbf{h}_{lk,l}
\mathbf{x}_{p,lk}^{(t-1)}
\right)^*
\mathbf{h}_{lk,l}
\mathbf{x}_{p,lk}
\right\}
}
\\
+
{\sum_{i\neq l}
f_{p,lk}d_{lk}\mathfrak{R}\!\left\{\!\left(\mathbf{h}_{lk,i}
\mathbf{x}_{c,i}^{(t-1)}\right)^*\!
\mathbf{h}_{lk,i}
\mathbf{x}_{c,i}\!
\right\}
}
\\
+{
\sum_{[ij]\neq [lk]}f_{p,lk}d_{lk}\mathfrak{R}\!\left\{\!\left(\mathbf{h}_{lk,i}
\mathbf{x}_{p,ij}^{(t-1)}\right)^*\!
\mathbf{h}_{lk,i}
\mathbf{x}_{p,ij}\!
\right\}
}
\\
-
b_{p,lk}d_{lk}
{
\left(
\sum_{ij}\left|\mathbf{h}_{lk,i}
\mathbf{x}_{p,ij}\right|^2
+
\sum_{i\neq l}\left|\mathbf{h}_{lk,i}
\mathbf{x}_{c,i}\right|^2
\right),
}
\end{multline}
\begin{multline}\label{eq=lem=q--}
\!\!\!\bar{r}_{c,lk}\!\!\geq\!\! \tilde{r}_{c,lk}\!\!=\!
a_{c,lk}
+
2d_{lk}\mathfrak{R}
\left\{
\left(
\mathbf{h}_{lk,l}
\mathbf{x}_{c,l}^{(t-1)}
\right)^*
\mathbf{h}_{lk,l}
\mathbf{x}_{c,l}
\right\}
\\
+\sum_{ij}e_{lk}f_{c,lk}\mathfrak{R}\!\left\{\!\left(\mathbf{h}_{lk,i}
\mathbf{x}_{p,ij}^{(t-1)}\right)^*\!
\mathbf{h}_{lk,i}
\mathbf{x}_{p,ij}\!
\right\}
\\
+
\sum_{i\neq l}e_{lk}f_{c,lk}\mathfrak{R}\!\left\{\!\left(\mathbf{h}_{lk,i}
\mathbf{x}_{c,i}^{(t-1)}\right)^*\!
\mathbf{h}_{lk,i}
\mathbf{x}_{c,i}\!
\right\}
\\
-
b_{p,lk}e_{lk}
\left(
\sum_{ij}\left|\mathbf{h}_{lk,i}
\mathbf{x}_{p,ij}\right|^2
+
\sum_{i}\left|\mathbf{h}_{lk,i}
\mathbf{x}_{c,i}\right|^2
\right),
\end{multline}
where $t$ is the number of the current iteration, $a_{c,lk}$, $b_{c,lk}$, $a_{p,lk}$, $b_{p,lk}$,  $c_{lk}$, $d_{lk}$, $e_{lk}$, $f_{p,lk}$, and $f_{c,lk}$
are constants, given by
\begin{align*}
a_{p,lk}&=\log\left(1+\gamma_{p,lk}^{(t-1)}\right)
-
\gamma_{p,lk}^{(t-1)}+(f_{p,lk}-b_{p,lk})d_{lk}\sigma^2
\\&\hspace{1cm}
-
\frac{Q^{-1}(\epsilon_p)}{\sqrt{n_p}}\left(\frac{\sqrt{V_{p,lk}^{(t-1)}}}{2}
+\frac{1}{\sqrt{V_{p,lk}^{(t-1)}}}\right),
\\
a_{c,lk}&=\log\left(1+\gamma_{c,lk}^{(t-1)}\right)
-
\gamma_{c,lk}^{(t-1)}+(f_{c,lk}-b_{c,lk})e_{lk}\sigma^2
\\&\hspace{1cm}
-
\frac{Q^{-1}(\epsilon_c)}{\sqrt{n_c}}\left(\frac{\sqrt{V_{c,lk}^{(t-1)}}}{2}
+\frac{1}{\sqrt{V_{c,lk}^{(t-1)}}}\right),
\\
b_{p,lk}&=\!\gamma_{p,lk}^{(t)}\!+\!
\frac{\zeta_{p,lk}^{(t-1)}f_{p,lk}}{2},\hspace{.2cm}
b_{c,lk}\!
=\!\gamma_{c,lk}^{(t)}\!+\!
\frac{\zeta_{c,lk}^{(t-1)}f_{c,lk}}{2}, 
\\
f_{p,lk}&=\frac{2Q^{-1}(\epsilon_p)}{\sqrt{n_pV_{p,lk}^{(t-1)}}},
\hspace{.2cm}
f_{c,lk}=\frac{2Q^{-1}(\epsilon_c)}{\sqrt{n_cV_{c,lk}^{(t-1)}}},
\\
c_{lk}&=\!\!\left(\!\!\sigma^2\!\!+\!\sum_{i\neq l}\!
|\mathbf{h}_{lk,i}
\mathbf{x}_{c,i}^{(t-1)}|^2\!\!
+\!\!\!\!
\sum_{[ij]\neq [lk]}
\!\!|\mathbf{h}_{lk,l}
\mathbf{x}_{p,lj}^{(t-1)}|^2\!\!\right)^{-1}\!\!\!\!\!\!,\,
\\
d_{lk}\!&
=\!\!\left(\!\!
\sigma^2\!\!+\!\sum_{ij}\!\left|\mathbf{h}_{lk,i}\mathbf{x}_{p,ij}^{(t-1)}\right|^2\!\!+\!\!\sum_{i\neq l}\!\left|\mathbf{h}_{lk,i}\mathbf{x}_{c,i}^{(t-1)}\right|^2\!\!\right)^{-1},
\\
e_{lk}&=\left(\sigma^2+\sum_{ij}\left|\mathbf{h}_{lk,i}\mathbf{x}_{p,ij}^{(t-1)}\right|^2
+\sum_{i}\left|\mathbf{h}_{lk,i}\mathbf{x}_{c,i}^{(t-1)}\right|^2\right)^{-1},
\end{align*}
where $\gamma_{c,lk}^{(t-1)}$, $V_{c,lk}^{(t-1)}$, $\gamma_{p,lk}^{(t-1)}$ and $V_{p,lk}^{(t-1)}$ are, respectively, obtained by replacing $\{\mathbf{x}^{(t-1)}\}$ in \eqref{sinr}, \eqref{dis-}, \eqref{sinr=p} and \eqref{dis-p}. Moreover, $\zeta_{p,lk}^{(t-1)}=d_{lk}c_{lk}^{-1}$ and $\zeta_{c,lk}^{(t-1)}=e_{lk}d_{lk}^{-1}$.
\end{lemma}
\begin{proof}
Please refer to Appendix \ref{ap-1}.
\end{proof}
Substituting $\tilde{r}_{p,lk}$s and $\tilde{r}_{c,lk}$s in $f_i$s gives the surrogate functions $\tilde{f}_i$s and consequently, the following surrogate optimization problem
\begin{subequations}\label{ar-opt-x-sur}
\begin{align}
\label{12-a-x-sur}
 \underset{\{\mathbf{x}\}\in\mathcal{X},\mathbf{r}_c
 }{\max}
& 
  \tilde{f}_0\left(\left\{\mathbf{x}\right\},\{\bm{\Theta}\}\right) &
  \\
 \,\,\,\,\,\,\,\,\,\,\,\, \,\, \text{s.t.}   \,\,\,\,\,\,\,\,\,\,&  \tilde{f}_i\left(\left\{\mathbf{x}\right\},\{\bm{\Theta}\}\right)\geq0,
\forall i,
\\
\label{4-c-x-sur}
 &
\tilde{r}_{lk}=\tilde{r}_{p,lk}+{r}_{c,lk}\geq r_{lk}^{th},
\forall l,k,
\\\label{4-b-x-sur}
 &  \sum_{k=1}^K\!\!r_{c,lk}\!\!\leq\! \min_{k}\!\left\{\tilde{r}_{c,lk}\right\}\!\!\triangleq\! \tilde{r}_{c,l}(\{\mathbf{x}\}),
\forall l,
\\
\label{4-d-x-sur}
&
r_{c,lk}\geq 0,
\forall l,k.
 \end{align}
\end{subequations}
This optimization problem is convex for spectral efficiency metrics, which can be efficiently solved by numerical tools.  
We can employ Dinkelbach-based algorithms to find the solution of \eqref{ar-opt-x-sur} for energy efficiency metrics such as GEE and weighted minimum EE.  Due to a space restriction, we do not provide the solutions here and refer the reader to \cite{zappone2015energy, soleymani2022rate} for more details.

\section{Proposed optimization framework for optimizing BD-RIS components}\label{seciv}
In this section, we solve \eqref{ar-opt} for fixed $\{\mathbf{x}^{(t)}\}$ to update the BD-RIS components. 
To this end, we employ an MM-based algorithm to obtain a suboptimal solution for the complicated optimization problem. 
For fixed $\{\mathbf{x}^{(t)}\}$, \eqref{ar-opt} is non-convex because of two reasons. First, the rates are not concave in $\{\bm{\Theta}\}$. Second, the set $\mathcal{T}_I$ is non-convex due to the constraint $\bm{\Theta}\bm{\Theta}^H={\bf I}$. To solve \eqref{ar-opt} for fixed $\{\mathbf{x}^{(t)}\}$, we have to handle these two issues. To this end, we first obtain suitable surrogate functions for the rates. Then we try to convexify $\mathcal{T}_I$. The rates have a similar structure with respect to the channels and beamforming vectors. Thus, we can employ an approach similar to Lemma \ref{lem-1} to find concave lower bounds for the rates with respect to $\{\bm{\Theta}\}$. In the following corollary, we state the concave lower bound for $\bar{r}_{c,lk}$. Due to a strict space restriction, we do not provide the concave lower bound for ${r}_{p,lk}$ since it is straightforward to obtain it from Lemma \ref{lem-1} and Corollary \ref{lem-2}.
\begin{corollary}\label{lem-2}
A concave lower bound for $\bar{r}_{c,lk}$ is
\begin{multline}\label{eq=lem=q-3}
\!\!\!\bar{r}_{c,lk}\!\!\geq\!\! \hat{r}_{c,lk}\!\!=\!
a_{c,lk}
+
2d_{lk}\mathfrak{R}
\left\{
\left(
\mathbf{h}_{lk,l}^{(t-1)}
\mathbf{x}_{c,lk}^{(t)}
\right)^*
\mathbf{h}_{lk,l}
\mathbf{x}_{c,lk}^{(t)}
\right\}
\\
+\sum_{ij}e_{lk}f_{c,lk}\mathfrak{R}\!\left\{\!\left(\mathbf{h}_{lk,i}^{(t-1)}
\mathbf{x}_{p,ij}^{(t)}\right)^*\!
\mathbf{h}_{lk,i}
\mathbf{x}_{p,ij}^{(t)}\!
\right\}
\\
+
\sum_{i\neq l}e_{lk}f_{c,lk}\mathfrak{R}\!\left\{\!\left(\mathbf{h}_{lk,i}^{(t-1)}
\mathbf{x}_{c,i}^{(t)}\right)^*\!
\mathbf{h}_{lk,i}
\mathbf{x}_{c,i}^{(t)}\!
\right\}
\\
-
b_{p,lk}e_{lk}
\left(
\sum_{ij}\left|\mathbf{h}_{lk,i}
\mathbf{x}_{p,ij}^{(t)}\right|^2
+
\sum_{i}\left|\mathbf{h}_{lk,i}
\mathbf{x}_{c,i}^{(t)}\right|^2
\right),
\end{multline}
where $\mathbf{h}_{lk,i}^{(t-1)}=\mathbf{h}_{lk,i}\left(\bm{\Theta}^{(t-1)}\right)$, and the other parameters are defined as in Lemma \ref{lem-1}.
\end{corollary}
Let us call the concave lower bound for ${r}_{p,lk}$ as $\hat{r}_{p,lk}$. 
Substituting $\hat{r}_{p,lk}$s and $\hat{r}_{c,lk}$s in $f_i$s yield the surrogate functions $\hat{f}_i$s as well as the following problem
\begin{subequations}\label{opt-t-sur}
\begin{align}
 \underset{\{\bm{\Theta}\}\in\mathcal{T},\mathbf{r}_c
 }{\max}
& 
  \hat{f}_0\left(\left\{\mathbf{x}\right\},\{\bm{\Theta}\}\right) &\\
 \,\,\,\,\,\,\,\,\,\,\,\, \,\, \text{s.t.}   \,\,\,\,\,\,\,\,\,\,&  \hat{f}_i\left(\left\{\mathbf{x}\right\},\{\bm{\Theta}\}\right)\geq0,
\forall i,
\label{4-b-t-sur}
\\
\label{4-c-t-sur}
 &
\hat{r}_{lk}=\hat{r}_{p,lk}+r_{c,lk}\geq r_{lk}^{th},
\forall l,k,
\\\label{4-d-t-sur}
 &  \sum_{k=1}^K\!\!r_{c,lk}\!\!\leq\! \min_{k}\!\left\{\hat{r}_{c,lk}\right\}\!\!\triangleq\! \hat{r}_{c,l}(\{\bm{\Theta}\}),
\forall lk,
\\
\label{4-e-t-sur}
&
r_{c,lk}\geq 0,
\forall l,k,
 \end{align}
\end{subequations}
which is convex only for $\mathcal{T}_U$. 
{
The proposed framework converges to a stationary point of \eqref{ar-opt} when $\mathcal{T}$ is a convex set.} 
Note that $\mathcal{T}_U$ contains the convex constraint $\bm{\Theta}_{m_g}\bm{\Theta}^H_{m_g}\preceq {\bf I}$, which may not be suitable for implementing in some existing numerical solvers. In the following, we propose a suboptimal approach to rewrite the constraint $\bm{\Theta}_{m_g}\bm{\Theta}^H_{m_g}\preceq {\bf I}$ as a series of inequality constraints on scalar optimization parameters, which are referred to as disciplinary convex constraints that can be easily handled in numerical solvers. Moreover, we propose an approach to convexify $\mathcal{T}_I$ in Section \ref{iv-a}.
\subsubsection{Making $\bm{\Theta}_{m_g}\bm{\Theta}^H_{m_g}\preceq {\bf I}$ a disciplinary convex constraint} 
The constraint $\bm{\Theta}_{m_g}\bm{\Theta}^H_{m_g}\preceq {\bf I}$ can be equivalently expressed as ${\bf T}={\bf I}-\bm{\Theta}_{m_g}\bm{\Theta}^H_{m_g}\succeq {\bf 0}$.
The matrix ${\bf T}$ can be written as
\begin{align}
\nonumber
{\bf T}&={\bf I}-\bm{\Theta}_{m_g}\bm{\Theta}^H_{m_g}={\bf I}-\bm{\Theta}_{m_g}\bm{\Theta}^*_{m_g}
\\&
=
\left[\begin{array}{cc}
1-|\theta_{11}|^2-|\theta_{12}|^2&-\theta_{11}^*\theta_{12}-\theta_{12}^*\theta_{22}\\
-\left(\theta_{11}^*\theta_{12}+\theta_{12}^*\theta_{22}\right)^*&1-|\theta_{12}|^2-|\theta_{22}|^2
\end{array}\right]
.
\end{align}
Since ${\bf T}$ is a $2\times 2$ Hermitian matrix, ${\bf T}$ is PSD if and only if the following constraints hold
\begin{align}\label{1-2}
|\theta_{11}|^2+|\theta_{12}|^2&\leq 1,\\
\label{2-2}
|\theta_{12}|^2+|\theta_{22}|^2&\leq 1,\\
\nonumber
\zeta_{m_g}=\left(1-|\theta_{11}|^2-|\theta_{12}|^2\right)\left(1-|\theta_{12}|^2-|\theta_{22}|^2\right)
&\\
-|\theta_{11}^*\theta_{12}+\theta_{12}^*\theta_{22}|^2&\geq 0.
\label{3-2}
\end{align}
Note that since $\bm{\Theta}_{m_g}$ is symmetric, we have $\bm{\Theta}_{m_g}^H=\bm{\Theta}_{m_g}^*$.
The constraints \eqref{1-2} and \eqref{2-2} are convex. 
Moreover, the constraint \eqref{3-2} can be simplified to
\begin{multline}\label{10}
\zeta_{m_g}=
\underbrace{1-2|\theta_{12}|^2-|\theta_{11}|^2-|\theta_{22}|^2-|\theta_{12}|^2|\theta_{11}^*+\theta_{22}|^2}_{\text{\small{Concave Part}}}
\\
\underbrace{+|\theta_{12}|^4\!\!+\!|\theta_{11}|^2|\theta_{22}|^2\!\!+\!|\theta_{12}|^2|\theta_{22}|^2\!\!+\!|\theta_{11}|^2|\theta_{12}|^2}_{\text{\small{Convex Part}}}\geq 0,
\end{multline}
which is not a convex constraint since $\zeta_{m_g}$ is not a concave function. 
However, we can apply CCP to convexify this constraint. That is, we keep the concave part of $\zeta_{m_g}$ and find a linear lower bound for the convex part of $\zeta_{m_g}$ by the first order Taylor expansion as
\begin{multline}\label{11}
\zeta_{m_g}\geq\hat{\zeta}_{m_g}=
1-2|\theta_{12}|^2-|\theta_{11}|^2-|\theta_{22}|^2
\\
-|\theta_{12}|^2|\theta_{11}^*+\theta_{22}|^2
+4|\theta_{12}^{(t-1)}|^2\mathfrak{R}\left(\theta^{(t-1)}_{12}\theta_{12}^*\right)-3|\theta_{12}^{(t-1)}|^4
\\
+2|\theta_{22}^{(t-1)}|^2\mathfrak{R}\left(\theta^{(t-1)}_{11}\theta_{11}^*\right)
+2|\theta_{11}^{(t-1)}|^2\mathfrak{R}\left(\theta^{(t-1)}_{22}\theta_{22}^*\right)
\\
-3|\theta_{11}^{(t-1)}|^2|\theta_{22}^{(t-1)}|^2
+2|\theta_{22}^{(t-1)}|^2\mathfrak{R}\left(\theta^{(t-1)}_{12}\theta_{12}^*\right)
\\
+2|\theta_{12}^{(t-1)}|^2\mathfrak{R}\left(\theta^{(t-1)}_{22}\theta_{22}^*\right)
-3|\theta_{12}^{(t-1)}|^2|\theta_{22}^{(t-1)}|^2
\\
+2|\theta_{11}^{(t-1)}|^2\mathfrak{R}\left(\theta^{(t-1)}_{12}\theta_{12}^*\right)
+2|\theta_{12}^{(t-1)}|^2\mathfrak{R}\left(\theta^{(t-1)}_{11}\theta_{11}^*\right)
\\
-3|\theta_{11}^{(t-1)}|^2|\theta_{12}^{(t-1)}|^2
\geq 0.
\end{multline}
Unfortunately, even though \eqref{11} is a convex constraint, it is not still a disciplinary constraint because of the convex function $|\theta_{12}|^2|\theta_{11}^*+\theta_{22}|^2$.
To address this issue, we employ the first order Taylor expansion to approximate $-|\theta_{12}|^2|\theta_{11}^*+\theta_{22}|^2$ as a linear function of $\theta_{11}$, $\theta_{12}$ and $\theta_{22}$, which yields $\tilde{\zeta}_{m_g}$. 
The constraint $\tilde{\zeta}_{m_g}\geq 0$ is a disciplinary convex constraint, which can be easily implemented in existing numerical solvers. 
Finally, by inserting the corresponding constraints into \eqref{opt-t-sur}, we have the following disciplinary convex problem
\begin{subequations}\label{opt-t-sur2-bd}
\begin{align}
 \underset{\{\bm{\Theta}\},
\mathbf{r}_c
 }{\max}
& 
  \hat{f}_0\left(\left\{\mathbf{x}\right\},\{\bm{\Theta}\}\right) \\
  \text{s.t.}   \,&\,  
\eqref{4-b-t-sur}-\eqref{4-e-t-sur},
\\
&
\tilde{\zeta}>1-\epsilon, \eqref{1-2},\eqref{2-2},\,\,\forall m,g.
 \end{align}
\end{subequations}
Let us call the solution of \eqref{opt-t-sur2-bd} as $\bm{\Theta}_{m_g}^{(\star)}$.
Unfortunately,  fulfilling the constraints $\tilde{\zeta}_{m_g}\geq 0$, \eqref{1-2} and \eqref{2-2} does not guarantee $\bm{\Theta}_{m_g}^{(\star)}\bm{\Theta}_{m_g}^{(\star)^H}\preceq {\bf I}$. 
To ensure obtaining a feasible point, we first check if $\bm{\Theta}_{m_g}^{(\star)}\bm{\Theta}_{m_g}^{(\star)^H}\preceq {\bf I}$ holds. If the largest eigenvalue of $\bm{\Theta}_{m_g}^{(\star)}\bm{\Theta}_{m_g}^{(\star)^H}$ is greater than 1, i.e., $\lambda_{m_g}>1$, then we choose $\hat{\bm{\Theta}}_{m_g}=\frac{\bm{\Theta}_{m_g}^{(\star)}}{\sqrt{\lambda_{m_g}}}$.
Finally,  we update $\bm{\Theta}_{m_g}$ as
\begin{equation}\label{eq-42-star}
\{\!\bm{\Theta}^{(t)}\!\}\!=\!\!
\left\{\!\!\!\!\!
\begin{array}{lcl}
\{\hat{\bm{\Theta}}\}\!\!&\!\!\!\!\!\!\text{if}\!\!&
f_0\left(\{\hat{\bm{\Theta}}\}\!\right)
>
f_0\left(\{\bm{\Theta}^{(t-1)}\}\!\right)
\\
\{\bm{\Theta}^{(t-1)}\}&&\text{Otherwise}.
\end{array}
\right.
\end{equation}
{
This updating rule ensures generating a sequence of non-decreasing $f_0$, which guarantees the convergence, but we do not make any claim on the optimality of the algorithm in this case.}

\subsection{Convexifying $\mathcal{T}_I$} \label{iv-a}
The constraints $\bm{\Theta}_{m_g}\bm{\Theta}^H_{m_g}={\bf I}$ and $\bm{\Theta}_{m_g}=\bm{\Theta}^T_{m_g}$ can be simplified as
\begin{align*}
\nonumber
\bm{\Theta}_{m_g}\bm{\Theta}^H_{m_g}
&\!
=\!
\left[\begin{array}{cc}
\theta_{11}&\theta_{12}\\
\theta_{12}&\theta_{22}
\end{array}\right]
\!
\left[\begin{array}{cc}
\theta_{11}^*&\theta_{12}^*\\
\theta_{12}^*&\theta_{22}^*
\end{array}\right]
\\&
\nonumber
\!=
\!
\left[\begin{array}{cc}
|\theta_{11}|^2+|\theta_{12}|^2&\theta_{11}^*\theta_{12}+\theta_{12}^*\theta_{22}\\
\left(\theta_{11}^*\theta_{12}+\theta_{12}^*\theta_{22}\right)^*&|\theta_{12}|^2+|\theta_{22}|^2
\end{array}\right]
\\&
\!=\!
\left[\begin{array}{cc}
1&0\\
0&1
\end{array}\right]\!,
\end{align*}
which results in
\begin{align}\label{1}
|\theta_{11}|^2+|\theta_{12}|^2&=1,\\
\label{2}
|\theta_{12}|^2+|\theta_{22}|^2&=1,\\
\label{3}
\theta_{11}^*\theta_{12}+\theta_{12}^*\theta_{22}&=0.
\end{align}
The constraints \eqref{1}-\eqref{3} are equivalent to $|\theta_{11}|=|\theta_{22}|$, and
\begin{align}\label{eq39}
\theta_{11}^*\theta_{12}\!+\theta_{12}^*\theta_{22}\!=\!|\theta_{12}|e^{j\angle \theta_{12}}\left(\theta_{11}^*+\theta_{22}e^{-j2\angle \theta_{12}} \right)=0,
\end{align}
which yields $\theta_{22}=\theta_{11}^*e^{j(2\angle \theta_{12}+\pi)}$. 
{
As can be verified through \eqref{1}-\eqref{eq39}, the phases and amplitudes of the coefficients for a group connected BD-RIS are highly dependent. For instance, if an amplitude of one of the coefficients is known, the amplitude of the other two coefficients can be uniquely obtained. Additionally, the phases are related to each other according to $\angle \theta_{22}=-\angle \theta_{11}+2\angle \theta_{12}+\pi$. }
The constraints \eqref{1} and \eqref{2} are not convex. To convexify \eqref{1}, we rewrite it as the following two constraints:
\begin{align}\label{1n}
|\theta_{11}|^2+|\theta_{12}|^2&\leq1,
\\
|\theta_{11}|^2+|\theta_{12}|^2&\geq1.
\label{1nn}
\end{align}
The constraint \eqref{1n} is a convex constraint, but \eqref{1nn} is not since $|\theta_{11}|^2+|\theta_{12}|^2$ is a jointly convex function in $|\theta_{11}|^2$ and $|\theta_{12}|^2$. Thus, we can employ CCP to convexify $|\theta_{11}|^2+|\theta_{12}|^2\geq1$
and relax it for a faster convergence as
\begin{multline}
\label{44}
|\theta_{11}^{{(t-1)}}|^2+2\mathfrak{R}\left(\theta_{11}^{{(t-1)}}(\theta_{11}-\theta_{11}^{{(t-1)}})^*\right)
+\!
|\theta_{12}^{{(t-1)}}|^2\!
\\
+\!
2\mathfrak{R}\!\left(
\!
\theta_{12}^{{(t-1)}}(\theta_{12}-\theta_{12}^{{(t-1)}})^*
\!
\right)
\!\!
\geq 
\!
1-\epsilon,
\end{multline}
where $\epsilon>0$. Similarly, we can convexify \eqref{2} by considering the following constraints:
\begin{align}
|\theta_{22}|^2+|\theta_{12}|^2&\leq1
\\
\nonumber
|\theta_{22}^{{(t-1)}}|^2+2\mathfrak{R}\left(\theta_{22}^{{(t-1)}}(\theta_{22}-\theta_{22}^{{(t-1)}})^*\right)
\\
+\!
|\theta_{12}^{{(t-1)}}|^2\!
+\!
2\mathfrak{R}\!\left(
\!
\theta_{12}^{{(t-1)}}(\theta_{12}-\theta_{12}^{{(t-1)}})^*
\!
\right)
\!\!
&\geq 
\!
1-\epsilon,
\label{43-n}
\end{align}
 In addition to \eqref{1} and \eqref{2}, the constraint \eqref{3} (or $\theta_{22}=\theta_{11}^*e^{j(2\angle \theta_{12}+\pi)}$) is not convex neither. A suboptimal way to convexify this constraint is to fix the phase of $\theta_{12}$.
For instance, if $\theta_{12}$ is real (or pure imaginary), then $\theta_{22}=-\theta_{11}^*$ (or $\theta_{22}=\theta_{11}^*$). 
Note that when $\theta_{12}=0$, the BD-RIS with group-connected architecture of group size two is equivalent to the diagonal RIS, and the constraint \eqref{3} is automatically satisfied. Thus, in this case, $\theta_{22}$ is independent of $\theta_{11}$, and there is no need to consider $\theta_{22}=-\theta_{11}^*$ (or $\theta_{22}=\theta_{11}^*$).
As a result, this algorithm never performs worse than the diagonal RIS.
Finally, the corresponding optimization problem in this case 
is
\begin{subequations}\label{opt-t-sur2-bd2}
\begin{align}
 \underset{\{\bm{\Theta}\},
\mathbf{r}_c
 }{\max}\,\,
& 
  \hat{f}_0\left(\left\{\mathbf{x}\right\},\{\bm{\Theta}\}\right) \\
  \text{s.t.}\,\,   \,&\,  
\eqref{4-b-t-sur}-\eqref{4-e-t-sur},
\\
&
\eqref{1n},\eqref{44}-\eqref{43-n},\,\,\forall m,g,
\\
&
\theta_{22}=-\theta_{11}^*\,\,\, \text{if}\,\,\,\theta_{12}\neq0,\,\,\forall m,g,
 \end{align}
\end{subequations}
which is convex and can be efficiently solved. 
{
Let us call the solution of \eqref{opt-t-sur2-bd2} as $\bm{\Theta}_{m_g}^{(\star)}$. 
We first normalize $\bm{\Theta}_{m_g}^{(\star)}$ to satisfy \eqref{1} and \eqref{2}, which results in $\hat{\bm{\Theta}}_{m_g}$. Finally, we update $\bm{\Theta}_{m_g}$ based on the rule in \eqref{eq-42-star}, which ensures the convergence, but we do not make any claim on the optimality for our framework with $\mathcal{T}_I$.
}

{
\subsection{Computational complexity analysis}

In this subsection, we provide computational complexity analysis only for the MWRM problem with the feasibility set $\mathcal{T}_U$ due to a strict page limitation. However, it can be straightforward to extend the analysis to other optimization problems with feasibility set $\mathcal{T}_I$. Note that the actual computational complexity of our algorithms depends on their implementation. Here, we provide an approximation for the number of multiplications to obtain a solution for our schemes.    

Our proposed framework is iterative, and in each iteration, two convex problems should be solved to update the beamforming vectors, RSMA parameters and BD-RIS coefficients. Firstly, we have to solve \eqref{ar-opt-x-sur} to update $\{{\bf x}\}$. According to \cite[Chapter 11]{boyd2004convex}, the number of Newton iterations to solve a convex optimization problem is proportional to the square root of the number of constraints in the problem. Thus, the number of Newton iterations to solve  \eqref{ar-opt-x-sur} grows with the square root of the total number of users in the system, i.e., $\sqrt{LK}$. To solve each Newton iteration, we have to obtain $2LK$ concave lower bounds for the rates. Each concave lower bound has a quadratic structure, and the number of multiplications to compute each is proportional to $N_{BS}LK$. Thus, we can approximate the computational complexity for solving \eqref{ar-opt-x-sur} by $\mathcal{O}\left(N_{BS}L^2K^2\sqrt{KL}\right)$. 
To update $\{\bm{\Theta}\}$, we have to solve \eqref{opt-t-sur2-bd}, which has $3KL+1.5MN_{RIS}$ constraints. To solve each iteration, we have to compute $LK$ channels, which approximately requires $MN_{RIS}^2N_{BS}$ multiplications per channel. 
Furthermore, we have to compute $2LK$ rates, which approximately requires $2N_{BS}L^2K^2$ multiplications ($N_{BS}LK$ multiplications for each rate).
Hence, the overall computational complexity of solving \eqref{opt-t-sur2-bd} can be approximated as $\mathcal{O}\left(LKN_{BS}\left(MN_{RIS}^2+2LK\right)\sqrt{3KL+1.5MN_{RIS}}\right)$. 
Considering $N$ as the maximum number of iterations for our algorithm, the overall computational complexity can be approximated as $N$ times the summation of the computational complexities of updating $\{{\bf x}\}$ and $\{\bm{\Theta}\}$.}

\section{Numerical results}\label{secnum}
In this section, we provide some numerical results for the MWRM and MWEEM problems. We consider $\epsilon_c=\epsilon_p$ and $n_p=n_c.$ 
{
We assume that there is a line of sight (LoS) link between each BS and each (BD-)RIS as well as between each (BD-)RIS and each user, similar to e.g., \cite{pan2020multicell, soleymani2022improper}. Thus, we assume that small-scale fading for the channels ${\bf f}_{lk,m}$ and ${\bf G}_{mi}$ for all $l,k,m,i$ follow Rician fading as described in \cite[Eq. (55)]{pan2020multicell}. We consider a Rician factor of $3$ for these channels \cite{pan2020multicell}. Additionally, we assume that there is a non-LoS link between each BS and each user, which results in a Rayleigh small-scale fading for ${\bf d}_{lk,i}$ for all $l,k,i$ \cite{pan2020multicell, soleymani2022improper}. 
For large-scale fading, we employ the well-known path loss model, given in e.g., \cite[Eq. (59)]{soleymani2022improper}. 
The other propagation parameters, including the antenna gains, bandwidth, noise power density, path loss components, and the path loss at the reference distance of $1$ meter, are based on \cite{soleymani2022improper, soleymani2023spectral}.}
 Moreover, the simulation scenario is a two-cell MISO BC, as depicted in \cite[Fig. 5]{soleymani2023spectral}, where the BSs/RISs/users positions are chosen similar to \cite{ soleymani2023spectral}. Additionally, the considered schemes in the simulations are represented as:
{\bf RS-BD-RIS$_{U}$} (or {\bf RS-BD-RIS$_{I}$}): The proposed scheme for BD-RIS-assisted systems with rate splitting, BD-RIS with group-connected architecture of group size two and feasibility set $\mathcal{T}_U$ (or $\mathcal{T}_I$).
{\bf RS-RIS$_{X}$} (or {\bf TIN-RIS$_{X}$}): The proposed algorithm for RIS-assisted systems with rate splitting (or TIN), regular RIS and feasibility set $\mathcal{T}_X$, where $X$ can be $U$ or $I$.
{\bf RS-RIS$_R$} (or {\bf TIN-RIS$_R$}): The proposed algorithm for RIS-assisted systems with rate splitting (or TIN), regular RIS and random RIS elements.
{\bf RS-No-RIS} (or {\bf TIN-No-RIS}): The RSMA (or TIN) scheme for systems without RIS.
{\bf Sh-RS-BD-RIS$_U$}: The proposed scheme for BD-RIS-assisted systems with rate splitting, BD-RIS with group-connected architecture of group size two, feasibility set $\mathcal{T}_U$ and considering Shannon rates.

\subsection{Fairness rate}
In this subsection, we present some numerical results for the fairness rate, which is defined as the  maximum of the minimum rate of users. We consider the impact of the reliability constraint, packet length, number of users per cell and power budget on the performance of RSMA and RIS. Additionally, we compare the performance of regular RIS with BD-RIS with group connected architecture of group size two.

\subsubsection{Impact of the reliability constraint}
\begin{figure}
    \centering
    \begin{subfigure}{0.45\textwidth}
        \centering
\includegraphics[width=\textwidth]{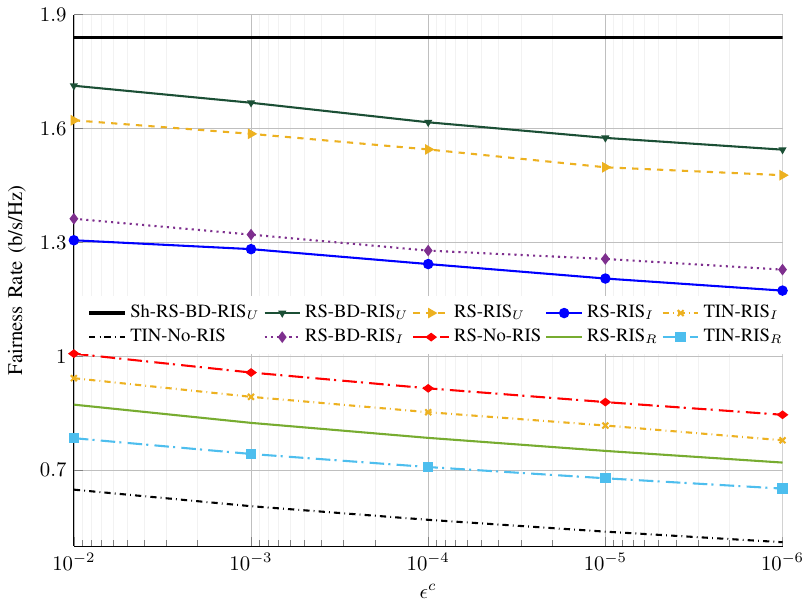}
        \caption{Average fairness rate.}
    \end{subfigure}
~
\begin{subfigure}{0.45\textwidth}
        \centering
\includegraphics[width=\textwidth]{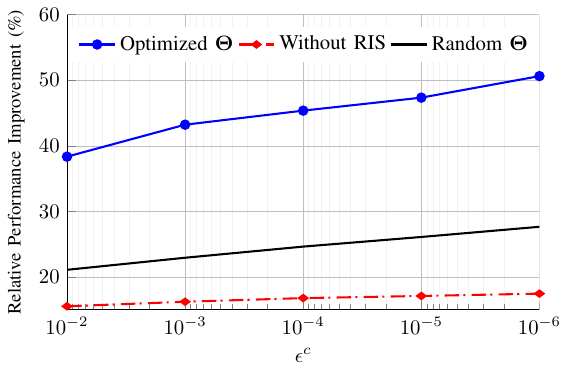}
        \caption{Average relative performance improvement by RSMA.}
    \end{subfigure}%
    \caption{The average fairness rate and performance improvement by RSMA versus $\epsilon^c$ for multi-cell MISO BC with $N_{RIS}=20$, $L=2$, $M=2$, $n_t=200$ bits,  $K=6$, $N_{BS}=5$ and $P=10$dB.}\vspace{-.6cm}
	\label{Fig-r-eps}  
\end{figure}
Fig. \ref{Fig-r-eps}  shows the average fairness rate and performance improvement by RSMA versus $\epsilon^c$ for multi-cell MISO  BC with $N_{RIS}=20$, $L=2$, $M=2$, $n_t=200$ bits,  $K=6$, $N_{BS}=5$ and $P=10$dB.
As expected, the fairness rate for all schemes decreases when the reliability constraint is more stringent. 
Moreover, we can observe that the RSMA scheme with BD-RIS with group-connected architecture of group size two and the feasibility set $\mathcal{T}_U$ can outperform the other schemes. Additionally, we observe that RIS (either regular or BD) can significantly improve the system performance and enhance the RSMA benefits. Note that in this example, the number of users per cell is higher than the number of BS antennas, which means that the considered system is overloaded. Interestingly, we can observe that RSMA without RIS can outperform TIN with regular RIS, which  shows the importance of employing an effective interference-management technique in overloaded systems.
Finally, we observe that the benefits of RSMA increase with ${\epsilon^c}^{-1}$. It shows that managing interference is even more important in URLLC systems. Furthermore, we observe that the RSMA benefits with employing RIS with properly optimized elements.

\subsubsection{Impact of the packet length}
\begin{figure}
    \centering
    \begin{subfigure}{0.45\textwidth}
        \centering
\includegraphics[width=\textwidth]{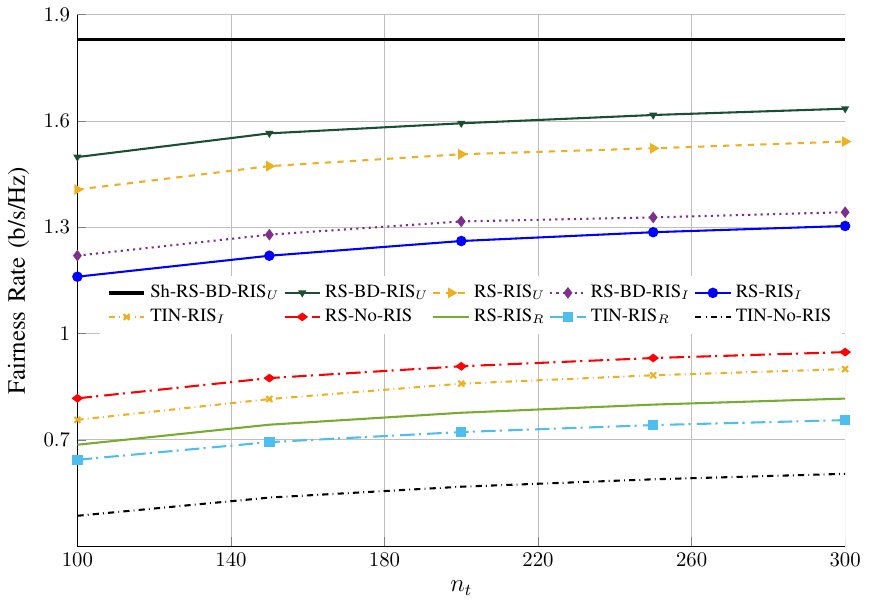}
        \caption{Average fairness rate.}
    \end{subfigure}
~
\begin{subfigure}{0.45\textwidth}
        \centering
\includegraphics[width=\textwidth]{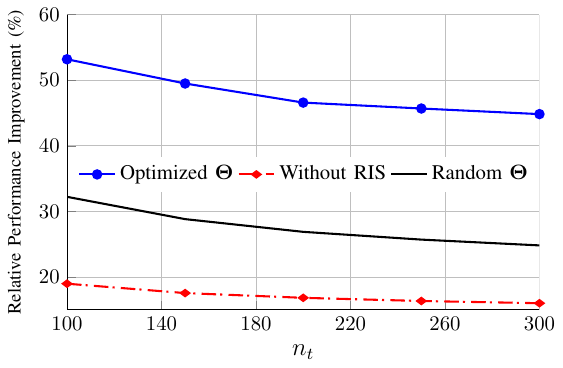}
        \caption{Average relative performance improvement by RSMA.}
    \end{subfigure}%
    \caption{The average fairness rate and performance improvement by RSMA versus $n_t$ for multi-cell MISO BC with $N_{RIS}=20$, $L=2$, $M=2$, $n_t=256$,  $K=6$, $\epsilon^c=0.001$ and $P=10$dB.}
	\label{Fig-r-nt}  
\end{figure}
Fig. \ref{Fig-r-nt}  shows the average fairness rate and performance improvement by RSMA versus $n_t$ for multi-cell MISO BC with $N_{RIS}=20$, $L=2$, $M=2$, $n_t=256$,  $K=6$, $\epsilon^c=0.001$ and $P=10$dB. As can be observed, the RSMA scheme with BD-RIS with group-connected architecture of group size two and the feasibility set $\mathcal{T}_U$ can outperform the other schemes. Moreover, we can observe that RIS and RSMA can significantly increase the average fairness rate. As expected, the average fairness rate increases with $n_t$ for all the schemes. However, the benefits of RSMA decrease with $n_t$. Indeed, when the packet lengths are shorter (or the latency constraint is more stringent), RSMA can provide higher gains, which indicates that RSMA can be even more effective in URLLC systems.

\subsubsection{Impact of the number of users per cell}
\begin{figure}
            \centering
\includegraphics[width=.42\textwidth]{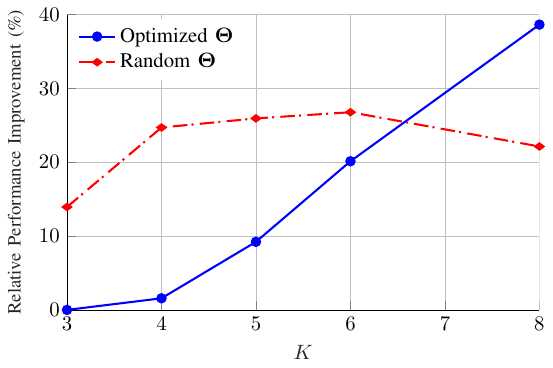}
            \caption{The average performance improvement by RSMA versus $K$ for multi-cell MISO RIS-assisted BC with $N_{RIS}=20$, $L=2$, $M=2$, $n_t=200$,  $\epsilon^c=0.001$ and $P=17$dB.}
	\label{Fig-r-ben}  
\end{figure}
Fig. \ref{Fig-r-ben} shows the impact of the number of users per cell as well as optimizing RIS components on the performance of RSMA. 
As can be observed, the benefits of RSMA can highly vary when optimizing RIS elements. Interestingly, we can observe that RIS may decrease RSMA benefits in underloaded systems, while it enhances the RSMA gain in overloaded systems.
This happens because  RIS can modulate the channels to mitigate the interference in underloaded systems. Thus,  a proper design of RIS elements may decrease the RSMA  benefits in underloaded systems. 
However,  as the number of users grows, the interference becomes more severe, and RIS alone cannot completely mitigate it. As a result, we observe that the RSMA benefits monotonically increase with $K$ when RIS elements are properly designed. 

\subsubsection{Impact of power budget}
\begin{figure}
    \centering
    \begin{subfigure}{0.45\textwidth}
        \centering
\includegraphics[width=\textwidth]{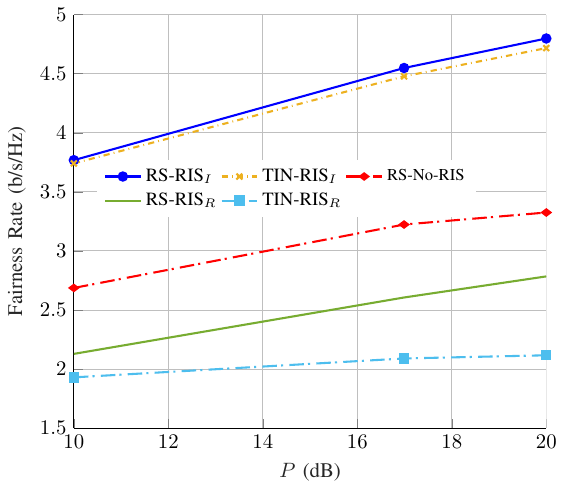}
        \caption{$K=4$.}
    \end{subfigure}
\begin{subfigure}{0.45\textwidth}
        \centering
\includegraphics[width=\textwidth]{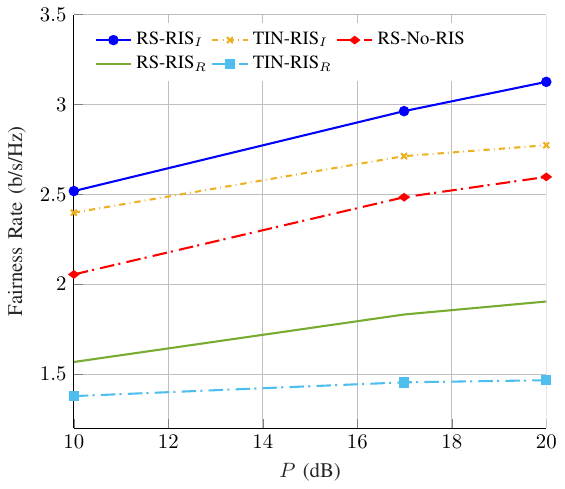}
        \caption{$K=5$.}
    \end{subfigure}%
    \caption{The average fairness rate versus the BSs power budget $P$ for $N_{BS}=8$, $N_{RIS}=20$, $L=2$, $M=2$, $n_t=200$,  $\epsilon^c=0.001$ and  $K$.}
	\label{Fig-r}  
\end{figure}
Fig. \ref{Fig-r} shows the average fairness rate versus the BSs power budget $P$ for $N_{BS}=8$, $N_{RIS}=20$, $L=2$, $M=2$, $n_t=200$,  $\epsilon^c=0.001$ and different $K$.
In this figure, the number of users per cell is less than the number of BS antennas, which is referred to as underloaded systems. As can be observed, RSMA cannot provide any considerable benefits in RIS-assisted systems for $K=4$. However, RSMA can significantly outperform TIN for $K=5$, especially in higher SNR regimes.  Indeed, it shows that RSMA can be still beneficial in underloaded systems. Additionally, we observe that RSMA can significantly outperform TIN with random RIS elements in underloaded scenarios. However, such benefits are lower (or may even vanish) when RIS elements are optimized.  
This result indicates that RIS can be employed as an interference-management technique in underloaded BCs. However, we should employ more advanced interference-management techniques to fully reap RIS benefits for overloaded systems, which is in line with the findings in \cite{soleymani2022rate, soleymani2022noma}.

\subsubsection{Comparison of RIS technologies}
In Fig. \ref{Fig-r-nt} and Fig. \ref{Fig-r-eps}, we compare the performance of various schemes, including regular and BD-RIS with different feasibility sets. In this part, we provide another comparison for regular RIS and BD-RIS with group-connected architecture of groups size two. 
\begin{figure}
    \centering
\includegraphics[width=.45\textwidth]{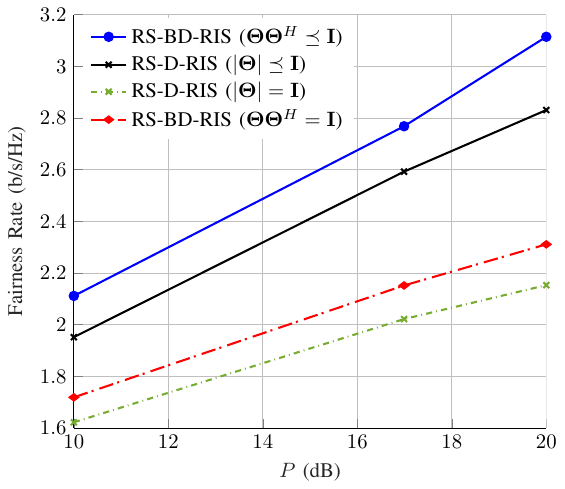}
            \caption{The average fairness rate versus power budget for $N_{RIS}=20$, $L=2$, $M=2$, $n_t=200$, $\epsilon^c=0.001$, and $N_{BS}=5$.}
	\label{Fig-bd}  
\end{figure}
Fig. \ref{Fig-bd} shows the average fairness rate versus the BSs power budget $P$ for $N_{BS}=5$, $N_{RIS}=20$, $L=2$, $M=2$, $n_t=200$, and $\epsilon^c=0.001$.
As can be observed, the BD-RIS with group-architecture of group size two and constraint $\bm{\Theta}\bm{\Theta}^H\preceq {\bf I}$ can highly outperform a regular diagonal RIS. 
As indicated in Section \ref{seciv}, BD-RIS is a more general architecture than regular RIS that includes regular RIS as a special case. Thus, BD-RIS never performs worse than regular RIS.

\subsection{Fairness EE}
In this subsection, we present some numerical results for the minimum EE of users, which is referred to as the fairness EE. 
\begin{figure}
    \centering
\includegraphics[width=.48\textwidth]{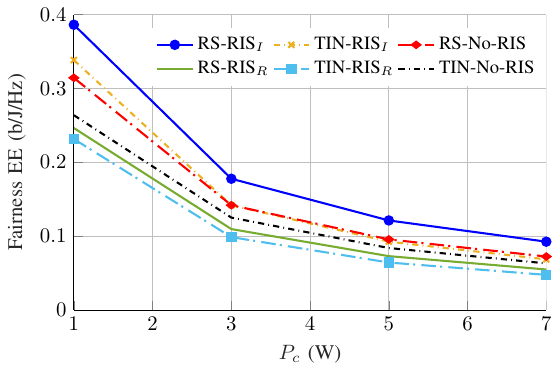}
    \caption{The average fairness EE versus $P_c$ for $N_{RIS}=20$, $L=2$, $M=2$, $n_t=200$, $\epsilon^c=0.001$, $N_{BS}=5$ 
and $K=7$.}
	\label{Fig-ee2}  
\end{figure}
In Fig. \ref{Fig-ee2}, we show the average fairness EE versus $P_c$ for $N_{RIS}=20$, $L=2$, $M=2$, $n_p=n_c=200$, $\epsilon_c=\epsilon_p=0.001$, $N_{BS}=5$ and  $K=7$. As can be observed, RIS can highly increase the fairness EE for both RSMA and TIN schemes when RIS components are properly optimized. We can also observe that RSMA can outperform TIN with and without RIS. In the following, we provide an in-depth discussion of the benefits of RIS and RSMA from an EE point of view.

\begin{figure}
    \centering
\includegraphics[width=.48\textwidth]{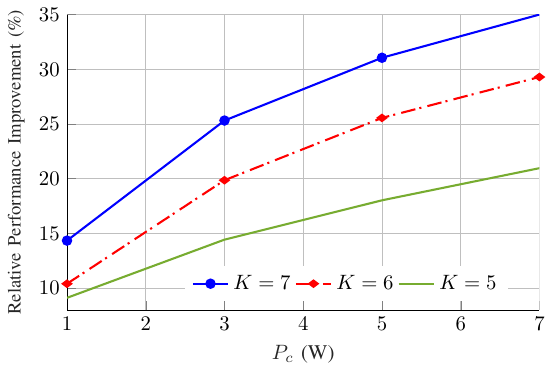}
            \caption{The average performance improvement by RSMA versus $P_c$ for $N_{RIS}=20$, $L=2$, $M=2$, $n_t=200$, $\epsilon^c=0.001$, and $N_{BS}=5$.}
	\label{Fig-ee}  
\end{figure}
In Fig. \ref{Fig-ee}, we show the average fairness EE improvement by RSMA for RIS-assisted systems. As can be observed, the benefits of employing RSMA increase with $K$. The reason is that, as the number of users grows, the interference level increases, which in turn improves the benefits of employing a powerful interference-management technique such as RS. Indeed, the more overloaded the systems is, the more gain RSMA can provide. We also observe that the RSMA benefits increase with $P_c$. Indeed, when the constant power consumption is higher, it is more important to properly design the system to get a better EE performance.

\begin{figure}[t!]
    \centering
\includegraphics[width=.48\textwidth]{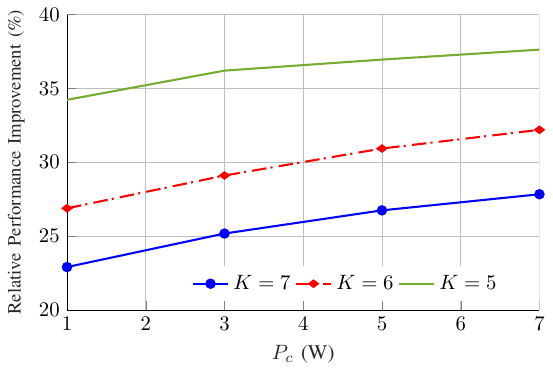}
            \caption{The average performance improvement by RIS for the RSMA schemes versus $P_c$ for $N_{RIS}=20$, $L=2$, $M=2$, $n_t=200$, $\epsilon^c=0.001$, and $N_{BS}=5$.}
	\label{Fig-ee3}  
\end{figure}
Fig. \ref{Fig-ee3} shows the average performance improvement by RIS for the RSMA schemes versus $P_c$ for $N_{RIS}=20$, $L=2$, $M=2$, $n_t=200$, $\epsilon^c=0.001$, and $N_{BS}=5$. As can be observed, the RIS benefits decrease with the number users for a fixed $N_{RIS}$. However, the RIS benefits are still significant even when the number of RIS elements per user is relatively low (less than 3 per user). This suggests that RIS may be promising in practical scenarios.  
We also observe that the RIS benefits increase with $P_c$, which is in line with the results in Fig. \ref{Fig-ee}.

\subsection{Summary}
 Our main findings in the numerical section can be summarized as follows:
 \begin{itemize}
 \item RSMA and RIS can significantly improve the SE and EE of the system. Moreover, the combination of the RSMA and BD-RIS with group-connected architecture of group size two outperforms the other schemes.

\item The use of RIS increases the benefits of RSMA in overloaded systems. However, the relative benefits of RSMA decrease by employing RIS in underloaded systems.

\item The benefits of RSMA increase with the number of users, $K$, since the interference level increases when there are more users in the system. However, the benefits of RIS decrease with $K$ when $N_{RIS}$ is fixed. The reason is that the number of RIS elements per user decreases with $K$ for a fixed $N_{RIS}$, which in turn reduces the RIS benefits.

\item RSMA provides higher gains when the packet lengths are shorter or when the reliability constraint is more stringent. Therefore, RSMA is more effective in URLLC systems, especially in highly overloaded systems.
 \end{itemize}

\section{{
Conclusion  and future work}}\label{seccon}
In this paper, we showed that RSMA and RIS can significantly improve the spectral and energy efficiency of MISO multi-cell BC URLLC systems. 
Moreover, we investigated the role of RSMA and RIS in URLLC systems, analyzing the impact on performance of different parameters, such as the reliability constraint, packet length, number of users per cell and BS power budget. We showed that the use of RIS has a different impact on the benefits of RSMA in different operational points. Specifically, RIS decreases RSMA benefits in underloaded systems, while it enhances the RSMA gain in overloaded systems. Indeed, RSMA and RIS can be mutually beneficial tools in overloaded systems. 
In addition, RSMA provides higher gains when the reliability constraint is more stringent and/or when the packet lengths are shorter. Finally, we showed that BD-RIS with group-connected architecture of group size two can outperform regular RIS. 

{
As a future work, it can be interesting to investigate the performance of multi-layer RSMA schemes in URLLC systems. Moreover, the performance gap between our proposed algorithms and the global optimal solution should be studied. 
Additionally, proposing robust designs for RSMA in multiple-antenna BD-RIS-assisted URLLC systems can be another challenging direction for future studies. In \cite{kim2022learning, peng2021robust, yang2021performance, hu2021robust}, the performance RIS has been studied in the presence of imperfect CSI, which might be helpful to obtain robust RSMA designs for RIS-assisted URLLC systems. Finally, there is no implementation of BD-RIS yet, and it is important to investigate the performance of BD-RIS with more realistic models based on practical implementations.}
\appendices
\section{Proof of Lemma \ref{lem-1}}\label{ap-1}\vspace{-.2cm}
Due to a space restriction, we only provide a proof for the inequality in \eqref{eq=lem=q--}. It is straightforward to extend the proof to obtain the lower bound $\tilde{r}_{p,lk}$ in \eqref{eq=lem=q}. To prove \eqref{eq=lem=q--}, we obtain a concave lower bound for the Shannon rate as well as a convex upper bound for $\delta_{c,lk}(\{\mathbf{x}\},\{\bm{\Theta}\})$. 
Employing \cite[Eq. (70)]{soleymani2023spectral}, we can obtain a lower bound for the Shannon rate as
\begin{multline}
\log\left(1+\gamma_{c,lk}\right)\geq
\log\left(1+\gamma_{c,lk}^{(t-1)}\right)
-
\gamma_{c,lk}^{(t-1)}
\\
+
2
d_{lk}\mathfrak{R}
\left\{
\left(
\mathbf{h}_{lk,l}
\mathbf{x}_{c,l}^{(t-1)}
\right)^*
\mathbf{h}_{lk,l}
\mathbf{x}_{c,l}
\right\}
\\
-
\gamma_{c,lk}^{(t)}e_{lk}
\left(
\sigma^2+
\sum_{ij}\left|\mathbf{h}_{lk,i}
\mathbf{x}_{p,ij}\right|^2
+
\sum_{i}\left|\mathbf{h}_{lk,i}
\mathbf{x}_{c,i}\right|^2
\right),
\end{multline}
which is quadratic and concave in $\{{\bf x}\}$. 
To obtain a convex upper bound for $\delta_{c,lk}(\{\mathbf{x}\},\{\bm{\Theta}\})$, we employ the inequality in \cite[Eq. (70)]{soleymani2023spectral}, which results in
\begin{multline}\label{pr-24}
\sqrt{V_{c,lk}}\leq 
\frac{\sqrt{V_{c,lk}^{(t)}}}{2}
+
\frac{\gamma_{c,lk}}{\sqrt{V_{c,lk}^{(t)}}\left(1+\gamma_{c,lk}\right)}
=\frac{\sqrt{V_{c,lk}^{(t)}}}{2}
\\+\frac{1}{\sqrt{V_{c,lk}^{(t)}}}\left(1
-
\frac
{\sigma^2+
\sum_{ij}\left|\mathbf{h}_{lk,i}
\mathbf{x}_{p,ij}\right|^2
+
\sum_{i\neq l}\left|\mathbf{h}_{lk,i}
\mathbf{x}_{c,i}\right|^2
}{
\sigma^2+
\sum_{ij}\left|\mathbf{h}_{lk,i}
\mathbf{x}_{p,ij}\right|^2
+
\sum_{i}\left|\mathbf{h}_{lk,i}
\mathbf{x}_{c,i}\right|^2
}\right).
\end{multline}
Unfortunately, the upper bound in the right-hand side of \eqref{pr-24} is not convex in $\{{\bf x}\}$.
To obtain a convex upper bound for the right-hand side of \eqref{pr-24}, we employ \cite[Eq. (69)]{soleymani2023spectral}, which yields 
\begin{multline}\label{pr-25}
\sqrt{V_{c,lk}}\leq \frac{\sqrt{V_{c,lk}^{(t)}}}{2}
+\frac{1}{\sqrt{V_{c,lk}^{(t)}}}\left(
1-2\sigma^2e_{lk}
\right.
\\
-2\sum_{ij}e_{lk}\mathfrak{R}\!\left\{\!\left(\mathbf{h}_{lk,i}
\mathbf{x}_{p,ij}^{(t-1)}\right)^*\!
\mathbf{h}_{lk,i}
\mathbf{x}_{p,ij}\!
\right\}
\\
-
2\sum_{i\neq l}e_{lk}\mathfrak{R}\!\left\{\!\left(\mathbf{h}_{lk,i}
\mathbf{x}_{c,i}^{(t-1)}\right)^*\!
\mathbf{h}_{lk,i}
\mathbf{x}_{c,i}\!
\right\}
\\
\left.+e_{lk}\zeta_{c,lk}^{(t-1)}
\left(
\sum_{ij}\left|\mathbf{h}_{lk,i}
\mathbf{x}_{p,ij}\right|^2
+
\sum_{i}\left|\mathbf{h}_{lk,i}
\mathbf{x}_{c,i}\right|^2
\right)
\right).
\end{multline}
Finally, we can obtain the lower bound in \eqref{eq=lem=q--} by substituting \eqref{pr-24} and \eqref{pr-25} in \eqref{rate}. 

\bibliographystyle{IEEEtran}
\bibliography{ref2}

\end{document}